\documentclass[sigconf,natbib,acmthm]{acmart} 

\newcommand{\compilefullversion}{true} 
\newcommand{\compilehidecomments}{true}

\definecolor{red}{HTML}{E51400} 
\definecolor{blue}{HTML}{0050EF} 
\definecolor{green}{HTML}{008A00} 
\definecolor{purple}{HTML}{AA00FF} 
\definecolor{orange}{HTML}{FF7F00}
\definecolor{gray}{HTML}{848482} 

\usepackage{booktabs} 

\renewcommand\footnotetextcopyrightpermission[1]{} 

\usepackage{hyperref}
\hypersetup{colorlinks, linkcolor=blue, citecolor=green, anchorcolor=blue, urlcolor=black}  
\usepackage{nicefrac}       
\usepackage{microtype}      
\usepackage{ifthen}
\usepackage{epstopdf}   
\usepackage[ruled,vlined,linesnumbered]{algorithm2e}
\usepackage{multirow}
\usepackage{thm-restate}
\setcitestyle{numbers,sort&compress}
\usepackage{sistyle}

\DeclareMathOperator*{\argmax}{argmax}

\newcommand{\E}{\mathbb{E}}
\newcommand{\I}{\mathbb{I}}
\newcommand{\R}{\mathbb{R}}

\newcommand{\bX}{\boldsymbol{X}}

\newcommand{\cC}{\mathcal{C}}
\newcommand{\cE}{\mathcal{E}}

\newcommand{\cR}{\mathcal{R}}
\newcommand{\cS}{\mathcal{S}}
\newcommand{\cM}{\mathcal{M}}
\newcommand{\cV}{\mathcal{V}}

\newcommand{\abs}[1]{\left| #1 \right|}

\newcommand{\dom}{{\rm dom}}
\newcommand{\avg}{{\rm avg}}

\newcommand{\rroot}{{\rm root}}
\newcommand{\ag}{{\rm ag}}
\newcommand{\ai}{{\rm ai}}

\newcommand{\OPT}{{\rm OPT}}

\newcommand{\WRGreedy}{{\sf WR-Greedy}}
\newcommand{\NodeSelection}{{\sf NodeSelection}}
\newcommand{\CRGreedy}{{\sf CR-Greedy}}
\newcommand{\CRNS}{\mbox{\sf CR-NS}}
\newcommand{\WRNS}{\mbox{\sf WR-NS}}

\newcommand{\CRIMM}{\mbox{\sf CR-NAIMM}}
\newcommand{\WRIMM}{\mbox{\sf WR-NAIMM}}
\newcommand{\WRIIMM}{\mbox{\sf WRI-NAIMM}}

\newcommand{\alg}[1]{\textnormal{\textsf{#1}}}

\newcommand{\newalg}[1]{{#1}}
\SetNlSty{bfseries}{\color{black}}{}

\SetCommentSty{mycommfont}

\ifthenelse{\equal{\compilefullversion}{false}}{%
	\newcommand{\OnlyInFull}[1]{}
	\newcommand{\OnlyInShort}[1]{#1}
}{%
	\newcommand{\OnlyInFull}[1]{#1}%
	\newcommand{\OnlyInShort}[1]{}%
}%

\ifthenelse{ \equal{\compilehidecomments}{true} }{%
    \newcommand{\wei}[1]{}
    \newcommand{\weiran}[1]{}
    \newcommand{\lichao}[1]{}
    \newcommand{\backup}[1]{}
}{
    \newcommand{\wei}[1]{{\color{blue!50!black}  [\text{Wei:} #1]}}
    \newcommand{\weiran}[1]{{\color{red!70!black} [\text{Weiran:} #1]}}
    \newcommand{\lichao}[1]{{\color{green!90!black} [\text{Lichao:} #1]}}
    \newcommand{\backup}[1]{{\color{green!50!black} #1}}
}

\setcopyright{none}

\begin{document}
    
\title{Multi-Round Influence Maximization \OnlyInFull{(Extended Version)}}

\author{Lichao Sun$^1$, Weiran Huang$^2$,  Philip S. Yu$^{1,2}$, Wei Chen$^{3,}$}
\authornote{Corresponding author}
\affiliation{$^1$University of Illinois at Chicago, $^2$Tsinghua University, $^3$Microsoft Research}
\affiliation{lsun29@uic.edu, huang.inbox@outlook.com, psyu@uic.edu, weic@microsoft.com}

%
%
%
%
%
%


\begin{abstract}
In this paper, we study
the Multi-Round Influence Maximization (MRIM) problem, where
influence propagates in multiple rounds independently from possibly different seed
sets, and the goal is to select seeds for each round to maximize the expected number of
	nodes that are activated in at least one round.
MRIM problem models the viral marketing scenarios in which advertisers conduct multiple rounds of 
viral marketing to promote one product.
We consider two different settings: 
    1) the non-adaptive MRIM, where the advertiser needs to determine the seed sets for all rounds at the very beginning, 
    and 2) the adaptive MRIM, where the advertiser can select seed sets adaptively based on the propagation results in the previous rounds.
For the non-adaptive setting, we design two algorithms that exhibit an interesting tradeoff
	between efficiency and effectiveness: 
	a cross-round greedy algorithm that selects seeds at a global level and achieves $1/2 - \varepsilon$
	approximation ratio, and a within-round greedy algorithm that selects seeds round by round and achieves
	$1-e^{-(1-1/e)}-\varepsilon \approx 0.46 - \varepsilon$ approximation ratio but saves running time by
	a factor related to the number of rounds. 
For the adaptive setting, we design an adaptive algorithm that guarantees $1-e^{-(1-1/e)}-\varepsilon$ 
	approximation to the adaptive optimal solution.
In all cases, we further design scalable algorithms based on the reverse influence sampling approach
	and achieve near-linear running time.
We conduct experiments on several real-world networks and demonstrate that our algorithms
	are effective for the MRIM task.
\end{abstract}

%
%



\maketitle

\sloppy

\section{Introduction}


Most companies need to advertise their products or brands on social networks,
through paying for influential people (seed nodes) on Twitter, with the hope that they can promote the products to their followers \cite{berger2014word}. 
The objective is to find a set of most influential people with limited budget for the best marketing effect.
Influence maximization (IM) is the optimization problem of 
	finding a small set of most influential nodes 
	in a social network that generates the largest influence spread.
It models viral marketing scenario in social networks~\cite{domingos01,richardson02,kempe03}, 
	and can also be applied to cascade detection~\cite{Leskovec07}, 
	rumor control~\cite{BAA11,HeSCJ12}, etc.
The standard IM problem and a number of its 
	variants has been extensively studied
	(c.f.~\cite{chen2013information}).
In most formulations, the IM is formulated as a one-shot task:
	the seed set is selected by the algorithm at the beginning, and one propagation
	pass from the seed set activates a subset of nodes in the network.
The objective is to maximize the expected number of activated nodes in this one propagation pass.
However, in many practical viral marketing scenarios, 
	an advertiser's viral marketing campaign may contain multiple rounds to promote
	one product.
Each round may be initiated from a different set of influential nodes.
The advertiser would like to maximize the total number of users adopting the
	product over all rounds.
	
We model the above scenario by the multi-round diffusion model and multi-round influence maximization (MRIM) task.
We consider the entire process over $T$ rounds.
In each round $t$, an independent diffusion is carried out starting from seed set $S_t$,
	and a random set of nodes, $I(S_t)$, is activated.
Then the total influence spread over $T$ rounds given seed sets $S_t, \ldots, S_T$, 
	$\rho(S_1, \ldots, S_T)$, is the expected 
	total number of activated nodes while considering all
	rounds together, namely $\rho(S_1, \ldots, S_T) = \E[|\bigcup_{t=1}^T I(S_t)|]$.
Note that a node activated in a previous round may be activated again and propagate
	influence to other nodes in a new round, but it will not be counted again in the
	final influence spread.
The MRIM task is to find seed sets $S_1, \ldots, S_T$, each of which has at most $k$ nodes,
	so that the final influence spread $\rho(S_1, \ldots, S_T)$ is maximized.
	
The MRIM task possesses some unique features different from the classical 
	IM task.
For example, in the classical IM, it makes no sense to select one seed node
	multiple times, but for MRIM, it may be desirable to select an influential node in multiple rounds
	to generate more influence.
Moreover, it enables adaptive strategies where the advertiser to adaptive select seeds in the next round
	based on the propagation results of previous rounds.
%

%

We study both the non-adaptive and adaptive versions of MRIM under the 
	Multi-Round Triggering (MRT) model.
For the non-adaptive MRIM problem, we design two algorithms that exhibit an interesting trade-off between efficiency and effectiveness.
The first cross-round greedy algorithm selects seeds globally cross different rounds, and achieves
	an approximation ratio of $1/2 - \varepsilon$, for any $\varepsilon > 0$.
The second within-round greedy algorithm selects seeds round by round and achieves
	an approximation ratio of $1-e^{-(1-1/e)}-\varepsilon \approx 0.46 - \varepsilon$.
The higher approximation ratio enjoyed by the cross-round greedy algorithm is achieved by
	investigating seed candidates in all rounds together in each greedy step, and thus incurs
	a higher running time cost at a factor proportional to the number of rounds.
%
For the adaptive MRIM problem, we rigorously formulate the problem
	according to the adaptive optimization framework specified in~\cite{GoloKrause11}.
We show that our formulation satisfies the {\em adaptive submodularity}
	defined in~\cite{GoloKrause11}.
Based on the adaptive submodularity, we design the \alg{AdaGreedy} algorithm
	that achieves $1-e^{-(1-1/e)}-\varepsilon$ approximation to the
	optimal adaptive policy.
%

For both the non-adaptive and adaptive cases, we greatly improve the scalability 
	 by incorporating the state-of-the-art reverse influence sampling (RIS) 
	 approach~\cite{BorgsBrautbarChayesLucier,tang14,tang15}.
In each case, the RIS method needs to be carefully revised to accommodate to the multi-round
	or adaptive situation.
In all cases, we show that our scalable algorithms could achieve near-linear running time
	with respect to the network size, greatly improving the corresponding Monte Carlo greedy algorithms.
%

To demonstrate the effectiveness of our algorithms, we conduct experiments
on real-world social networks, with both synthesized and learned influence parameters.
Our experimental results demonstrate that our algorithms are more effective than baselines, and
	our scalable algorithms run in orders of magnitude faster than their Monte Carlo greedy counterparts
	while keeping the influence spread at the same level.
The results also show some interesting findings, such as the non-adaptive cross-round algorithm could 
	achieve competitive influence spread as the adaptive greedy algorithm.
This may suggest that in practice one may need to consider whether spending the cost of collecting feedbacks
	and doing near-term adaptive strategies based on feedbacks, or spending the up front cost to
	do more global planning, and it opens new directions for further investigations.

%

To summarize, our contributions are:
	(a) proposing the study of both non-adaptive and adaptive MRIM problems;
	(b) proposing non-adaptive and adaptive greedy algorithms and showing their trade-offs;
	(c) designing scalable algorithms in both non-adaptive and adaptive settings; and
	(d) conducting experiments on real-world networks to demonstrate the effectiveness and the scalability of 
	our proposed algorithms.
\OnlyInShort{Due to space constraints, all proofs and supplementary materials are shown in the extended version
	\cite{sun2018multi}.}
\OnlyInFull{Due to space constraints, all proofs and supplementary materials are shown in the appendix.}

\subsection{Related Work}

Influence maximization is first studied by Domingos and 
	Richardson~\cite{domingos01,richardson02}, and then formulated
	as a discrete optimization problem by \citet{kempe03},
	who also formulate the independent cascade model, the linear threshold model, the
	triggering model, and provide a greedy approximation algorithm based on submodularity.
Since then, 
a significant number of papers studies improving the efficiency and scalability
	of influence maximization algorithms~\cite{ChenWY09,ChenYZ10,WCW12,JungHC12,BorgsBrautbarChayesLucier,tang14,tang15}. 
At this front, the state of the art is the reverse influence sampling (RIS) 
	approach~\cite{BorgsBrautbarChayesLucier,tang14,tang15}, and the \alg{IMM} algorithm
	of~\cite{tang15} is among the most competitive ones so far.
Our scalable algorithms are based on \alg{IMM}, but require careful redesigns in the reverse sampling method.
Other studies look into different variants, such as community, competitive and complementary
	influence maximization~\cite{kempe07,BAA11,chen2011influence,HeSCJ12,lu2015competition,he2016joint},
	adoption maximization~\cite{bhagat_2012_maximizing}, 
	robust influence maximization~\cite{ChenLTZZ16,HeKempe16}, etc.

Another related work is adaptive seeding \cite{seeman2013adaptive}, which uses the first-stage nodes and their accessible neighbors together as the seed set to maximize the influence, and is quite different from ours.
In terms of the multi-round diffusion model and influence maximization,
    \citet{lin2015learning} focus on the multi-party influence maximization where there must be at least two parties to compete in networks.
\citet{LeiMMCS15} use the same formulation of the multi-round diffusion
	model and the influence maximization objective as in our paper.
They focus on the online learning aspect of learning edge probabilities, 
while we study the offline
	non-adaptive and adaptive maximization problem when the edge probabilities are known.
Without a rigorous study of such offline problems, it is very difficult to
	assess the performance of online learning algorithms, and as a result
	they could only propose heuristic learning algorithms without any 
	theoretical guarantee. 
From this perspective, our study fills this important gap
	by providing a solid theoretical understanding of the offline multi-round influence
	maximization.

There are also a number of studies on the influence maximization bandit problem
	\cite{CWYW16,VaswaniL15,Wen2016,VKWGLS17}.
Their formulations also have multi-rounds, but each round has a separate influence maximization
	instance, and the total reward is a simple count of activated nodes in all rounds, so
	one node activated in multiple rounds will be counted multiple times.
This makes it qualitatively different from our formulation.
Moreover, these study focus on the online learning aspects of such repeated influence maximization
	tasks.
	
The adaptive MRIM study follows the adaptive optimization framework
	defined by \citet{GoloKrause11}.
They also study adaptive influence maximization as an application,
	but the adaptation is at per-node level: seeds are selected one by one. 
Later seed 
	can be selected based on the activation results from the earlier seeds, but the earlier seeds would not 
	help propagation again for later seeds.
This makes it different from our multi-round model.
\section{Model and Problem Definition} \label{sec:model}


\subsection{Multi-Round Diffusion Model}
In this paper, we focus on the well-studied {\em triggering model} \cite{kempe03}
as the basic diffusion model. 
A social network is modeled as a directed graph $G=(V,E)$, 
where $V$ is a finite set of vertices or nodes, 
and $E \subseteq V \times V$ is the set of 
directed edges connecting pairs of nodes.
The diffusion of information or influence proceeds in discrete time steps $\tau = 0, 1, 2, \dots$. 
At time $\tau=0$, the {\em seed set} $S_0$ is selected to be active, and also 
each node $v$ independently chooses a random {\em triggering set} $T(v)$ 
according to some distribution over subsets of its in-neighbors. 
At each time $\tau\ge 1$, an inactive node $v$ becomes active if at least 
one node in $T(v)$  is active by time $\tau-1$. 
The diffusion process ends when there is no more nodes activated in a time step.
We remark that the classical Independent Cascade (IC) model is a special case of the triggering model.
In the IC model, 
    every edge $(u,v) \in E$ is associated with a probability $p_{uv} \in [0, 1]$, 
    and $p_{uv}$ is set to zero if $(u,v) \notin E$.
Each triggering set $T(v)$ is generated by independently sampling $(u,v)$ with probability $p_{uv}$
and including $u$ in $T(v)$ if the sampling of $(u,v)$ is successful.

The triggering model can be equivalently described as propagations in {\em live-edge graphs}.
Given a class of triggering sets $\{T(v)\}_{v\in V}$, we can construct the live-edge graph
$L = (V, E(L))$, where $E(L) = \{(u,v) \mid v\in V, u\in T(v)\}$, and
each edge $(u,v)\in L$ is called a {\em live edge}.
It is easy to see that the propagation in the triggering model is the same as the deterministic
propagation in the corresponding live-edge graph $L$: if $A$ is the set of active nodes
at time $\tau-1$, then all directed out-neighbors of nodes in $A$ will be active
by time $\tau$.
An important metric in any diffusion model is the {\em influence spread},
    the expected number of active nodes when the propagation from the given  
    seed set $S_0$ ends, denoted as $\sigma(S_0)$.
Let $\Gamma(G,S)$ denote the set of nodes in graph $G$ that can be reached from the node set $S$.
Then, by the above equivalent live-edge graph model, we have
$\sigma(S_0) = \E[|\Gamma(L,S_0)|]$, where the expectation is taken over the distribution
of live-edge graphs. 

A set function $f:V\rightarrow \R$ is called {\em submodular} if
for all $S\subseteq O \subseteq V$ and $u \in V\setminus T$, 
$f(S\cup \{u\}) - f(S) \ge f(O \cup \{u\}) - f(O)$.
Intuitively, submodularity characterizes the diminishing return property often occurring
in economics and operation research.
Moreover, a set function $f$ is called {\em monotone} if for all $S\subseteq O \subseteq V$,
$f(S) \le f(O)$.
It is shown in~\cite{kempe03} that influence spread $\sigma$ for the triggering model
is a monotone submodular function .
A non-negative monotone submodular function allows a greedy solution to
its maximization problem subject to a cardinality constraint, with an
approximation ratio $1-1/e$, where $e$ is the base of the natural logarithm~\cite{NWF78}.
This is the technical foundation for most influence maximization tasks.

We are now ready to define the Multi-Round Triggering (MRT) model.
The MRT model includes $T$ {\em independent} rounds of influence diffusions.
In each round $t \in [T]$, the diffusion starts from a separate seed set $S_t$ with up to $k$ nodes,
    and it follows the dynamic in the classical triggering model described previously.
Since one node can be repeatedly selected as the seed set in different rounds, 
    to clarify the round, we use pair notation $(v, t)$ to denote a node $v$ in the seed set of round $t$.
We use $\cS_t = \{(v,t)\mid v\in S_t \}$ to represent the seed set of round $t$ 
	in the pair notation.
Henceforth, we always use the calligraphic $\cS$ for sets in the pair notation and the normal $S$ for node sets.
By the equivalence between the triggering model and the live-edge graph model, 
    the MRT model can be viewed as $T$ independent propagations in the $T$ live-edge graphs $L_1, L_2, \dots, L_T$, 
    which are drawn independently from the same distribution based on the triggering model. 
The total active nodes in $T$ rounds is counted by $\left|\bigcup_{t=1}^T \Gamma(L_t, S_t)\right|$, where $\Gamma(L_t, S_t)$ is the set of final active nodes in round $t$.
Given a class of seed set (in pair notation) $\cS:=\bigcup_{t=1}^T \cS_t$, the {\em influence spread} $\rho$ in the MRT model is
defined as 
\begin{equation*}
\rho(\cS)=\rho(\cS_1 \cup \cS_2 \cup \dots\cup \cS_T )=\E\left[\abs{\bigcup_{t=1}^T \Gamma(L_t, S_t)}\right],
\end{equation*}
where the expectation is taken over the distribution of all live-edge graphs 
$L_1, L_2, \ldots, L_T$.

%

\subsection{Multi-Round Influence Maximization}

The classical influence maximization problem is to choose a seed set $S_0$ of size at most $k$
to maximize the influence spread $\sigma(S_0)$. 
For the Multi-Round Influence Maximization (MRIM) problem, the goal is to select at most $k$ seed nodes of each
round, such that the influence spread in $T$ rounds is maximized.
We first introduce its non-adaptive version formally defined as follows.

\begin{definition}
	The {\em non-adaptive} Multi-Round Influence Maximization (MRIM) under the MRT model
	is the optimization task where
	the input includes the directed graph $G=(V,E)$, the triggering set distribution for every node
	in the MRT model, the number of rounds $T$, and each-round budget $k$, and the goal
	is to find $T$ seed sets $\cS_1^*$, $\cS_2^*$, \dots, $\cS_T^*$ with each seed set having at most
	$k$ nodes, such that the total influence spread is maximized, i.e.,
	\begin{equation*}
		\cS^* = \cS_1^* \cup \cS_2^* \dots\cup \cS_T^* = \argmax_{\cS\colon |S_t| \leq k, \forall t\in[T]} \rho(\cS).
	\end{equation*}
\end{definition}
The non-adaptiveness refers to the definition that one needs to find $T$ seed sets all at once
before the propagation starts.
In practice, one may be able to observe the propagation results in previous rounds and
select the seed set for the next round based on the previous results to increase the influence spread.
This leads to the adaptive version.
To formulate the {\em adaptive} MRIM requires the setup of the 
adaptive optimization framework, and we defer to Section~\ref{sec:AMRIMDef} for
its formal definition.

Note that as the classical influence maximization is NP-hard and is a special case of MRIM with
$T=1$, both the non-adaptive and adaptive versions of MRIM are NP-hard.
\section{Non-Adaptive MRIM} \label{sec:nonadaptive}


Let $\cV_t = \{(v,t) \mid v \in V\}$ be the set of all possible nodes in round $t$ (e.g., $\cS_t \subseteq \cV_t$), and
$\cV := \bigcup_{t = 1}^T \cV_t $.
We first show that the influence spread function $\rho$ for the MRT model is monotone and submodular.

\begin{restatable}{lemma}{mrtsubmodular} \label{lem:mrtsubmodular}
    Influence spread $\rho(\cS)$ for the MRT model satisfies
    (a) monotonicity: for any $\cS^A \subseteq \cS^B \subseteq \cV$,  $\rho(\cS^A)\leq\rho(\cS^B)$;
    and (b) submodularity: 
    for any $\cS^A \subseteq \cS^B \subseteq \cV$ and any pair $(v, t) \in \cV\setminus \cS^B$, 
    $\rho(\cS^A \cup \{(v, t)\})-\rho(\cS^A) \geq \rho(\cS^B \cup \{(v, t)\})-\rho(\cS^B)$.
\end{restatable}
\begin{proof}
	The proof of monotonicity is straightforward, so we next consider submodularity.
	According to the definition of the influence spread, we have
	\begin{align*}
		&\rho(\cS^A \cup \{(v,t)\})-\rho(\cS^A) \geq \rho(\cS^B \cup \{(v,t)\})-\rho(\cS^B) \\
		\Leftrightarrow & \E\left[\left|\bigcup_{i=1}^T \Gamma(L_i,S_i^A) \setminus \Gamma(L_t,S_t^A) \cup \Gamma(L_t,S_t^A \cup \{v\})\right|\right] \\
		&- \E\left[\left|\bigcup_{i=1}^T \Gamma(L_i,S_i^A)\right|\right] \\
		\Leftrightarrow & \E\left[\left|\bigcup_{i=1}^T \Gamma(L_i,S_i^A) \cup \Gamma(L_t,S_t^A \cup \{v\})\right|\right] - \E\left[\left|\bigcup_{i=1}^T \Gamma(L_i,S_i^A)\right|\right] \\
		\geq &  \E\left[\left|\bigcup_{i=1}^T \Gamma(L_i,S_i^B) \setminus \Gamma(L_t,S_t^B) \cup \Gamma(L_{t+1},S_t^A \cup \{v\})\right|\right] \\
		&- \E\left[\left|\bigcup_{i=1}^T \Gamma(L_i,S_i^B)\right|\right] \\
		\Leftrightarrow & \E\left[\left|\bigcup_{i=1}^T \Gamma(L_i,S_i^B) \cup \Gamma(L_t,S_t^B \cup \{v\})\right|\right] - \E\left[\left|\bigcup_{i=1}^T \Gamma(L_i,S_i^A)\right|\right]. \\
	\end{align*}
	Then it is sufficient to show that
	\begin{align}\label{proofeq:1}
		& \bigcup_{i=1}^T \Gamma(L_i,S_i^A) \cup \Gamma(L_{t},S_t^A \cup \{v\}) \setminus \bigcup_{i=1}^T \Gamma(L_i,S_i^A) \nonumber\\
		& \supseteq  \bigcup_{i=1}^T \Gamma(L_i,S_i^B) \cup \Gamma(L_{t},S_t^B \cup \{v\}) \setminus \bigcup_{i=1}^T \Gamma(L_i,S_i^B).
	\end{align}
	
	For a node $u \in \bigcup_{i=1}^T \Gamma(L_i,S_i^B) \cup \Gamma(L_{t},S_t^B \cup \{v\}) \setminus \bigcup_{i=1}^T \Gamma(L_i,S_i^B)$, 
	$u$ is reachable from $S_t^B \cup \{v\}$ in $L_{t}$, 
	but not reachable from $S_{i}^B$ in $L_{i}$ for any $i\in[T]$.
	Thus $u$ is also not reachable from $S_{i}^A$ in $L_{i}$ for any $i\in[t]$.
	Therefore, we conclude that $u \in \bigcup_{i=1}^T \Gamma(L_i,S_i^A) \cup \Gamma(L_{t},S_{t}^A \cup \{v\}) \setminus \bigcup_{i=1}^T \Gamma(L_i,S_i^A)$.
	Thus the submodularity holds.
\end{proof}

%
The monotonicity and submodularity above are the theoretical basis of designing and analyzing greedy algorithms for the non-adaptive MRIM.
In the following sections, we will consider two different settings separately for seed node selection: within-round and cross-round.
For the within-round setting, 
    one needs to determine the seed sets round by round,
    while for the cross-round setting,
    one is allowed to select nodes crossing rounds.
    
\subsection{Cross-Round Setting}

    \begin{algorithm}[t]
    \caption{{\CRGreedy}: Cross-Round Greedy Algorithm} \label{alg:crgreedy}
    \KwIn{Graph $G=(V,E)$, integers $T$, $k$ and $R$, triggering set distributions.}
    \KwOut{$\cS^o $.}
    $\cS^o \leftarrow \emptyset$; $\cC\leftarrow \cV$\;
    $c_1, c_2, \dots, c_t \leftarrow 0$\tcp*{node counter for each round}
    \For{$i = 1, 2, \dots, kT$}{ 
        for all $(v,t)\in \cC\setminus \cS^o$, 
        estimate ${\rho}(\cS^o \cup \{(v,t)\})$ by simulating the diffusion process $R$ times\;
        $(v_i, t_i) \leftarrow \argmax_{(v,t)\in \cC\setminus \cS^o} \hat{\rho}(\cS^o \cup \{(v, t)\} )$\; \label{line:crsimulate}
        $\cS^o \leftarrow \cS^o \cup \{(v_i, t_i)\}$;   
        $c_{t_i} \leftarrow c_{t_i}  + 1$\; 
        \If(\tcp*[f]{budget for round $t_i$ exhausts}){$c_{t_i}  \ge k$}{ $\cC \leftarrow \cC \setminus \cV_{t_i}$\;}
    }
    \textbf{return} $\cS^o$.
\end{algorithm}


We design a greedy algorithm for the non-adaptive MRIM under cross-round setting, 
named \alg{CR-Greedy} (Algorithm~\ref{alg:crgreedy}).
The idea of \alg{CR-Greedy} is that at every greedy step, it searches all $(v,t)$
	in the candidate space $\cC$ and
	picks the one having the maximum marginal influence spread without replacement.
If the budget for some round $t$ exhausts, then the remaining nodes of $\cV_t$ are removed from $\cC$.
Note that as $\cC$ contains nodes assigning to different rounds, \alg{CR-Greedy} selects nodes crossing rounds.

Given a set $U$ which is partitioned into disjoint sets $U_1, \dots, U_n$ and $\mathcal{I}=\{X\subseteq U\colon |X\cap U_i|\le k_i, \forall i\in[n]\}$, $(U,\mathcal{I})$ is called a {\em partition matroid}.
Thus, the node space $\cV$ with the constraint of MRIM, namely $(\cV, \{\cS\colon|S_t|\le k, \forall t\in[T]\})$, is a partition matroid.
This indicates that MRIM under cross-round setting is an instance of submodular maximization under partition matroid, 
and thus the performance of \alg{CR-Greedy} has the following guarantee \cite{fisher1978analysis}.

\begin{restatable}{theorem}{crmrim} \label{thm:crmrim}
    Let $\cS^*$ be the optimal solution of the non-adaptive MRIM under cross-round setting.
    For every $\varepsilon > 0$ and $\ell > 0$, with probability at least
    $1-\frac{1}{n^\ell}$, the output $\cS^o$ of  \alg{CR-Greedy}
    satisfies
    \begin{equation*}
    \rho(\cS^o) \geq \left(\frac{1}{2} - \varepsilon\right) \rho(\cS^*),
    \end{equation*}
    if \alg{CR-Greedy}
    uses $R=\lceil 31 k^2T^2n\log(2kn^{\ell+1})/\varepsilon^2 \rceil$ as input.
    In this case, the total running time is $O(k^3\ell T^4 n^2m\log(n)/\varepsilon^2)$, 
    assuming each simulation finishes in $O(m)$ time.
\end{restatable}

\subsection{Within-Round Setting}

%

We give the second greedy algorithm  (Algorithm~\ref{alg:wrgreedy}) for the within-round setting,
	denoted as \alg{WR-Greedy}.
The idea of \alg{WR-Greedy} is that seed nodes are selected round by round.
More specifically, we greedily select seed nodes for the first round, and only after we selected $k$ seed
	nodes for the first round, then we greedily select seed nodes for the next round, and so on.
The immediate advantage of \alg{WR-Greedy} over \alg{CR-Greedy} is that in each greedy step
	the \alg{WR-Greedy} only searches candidates $(v,t)$ with $t$ being the current round number, while
	\alg{CR-Greedy} needs to search $(v,t)$ for all rounds.
This would give at least a factor of $T$ saving on the running time of \alg{WR-Greedy}.
However, as we will show below, it pays a price of a slightly lower approximation ratio.
	
To analyze \alg{WR-Greedy}, we utilize the result of Lemma~\ref{lem:mrtsubmodular}
	in a different way.
First, when we fix the seed sets in round $1, \ldots, t-1$, and only vary the seed sets in round $t$,
	the influence spread certainly still satisfies the monotonicity and submodularity.
Therefore, within round $t$, the seed set $\cS^o_t$ selected by \alg{WR-Greedy} for round $t$ is a $(1-1/e-\varepsilon)$-approximate
	solution when the previous $t-1$ seed sets are fixed.
Second, we could view seed set $\cS_t$ of each round $t$ as a unit, and when adding it to the previous units,
	it would also satisfy the monotonicity and submodularity by Lemma~\ref{lem:mrtsubmodular}.
Namely, $\rho(\cS_1 \cup  \ldots \cup \cS_{t}) \le \rho(\cS_1 \cup \ldots \cup \cS_{t'}) $ for all $t < t'$, and 
	$\rho(\cS_1 \cup  \ldots \cup \cS_{t} \cup \cS_{t''}) - \rho(\cS_1 \cup  \ldots \cup \cS_{t})
	\le \rho(\cS_1 \cup  \ldots \cup \cS_{t'} \cup \cS_{t''}) - \rho(\cS_1 \cup  \ldots \cup \cS_{t'}) $,
	for all $t < t' < t''$.
This means that \alg{WR-Greedy} can be viewed as greedily selecting seed set units $\cS_t$ round by round
	with the monotonicity and submodularity, while within each round, it can not find the optimal $\cS_t$ but instead
	by a $(1-1/e-\varepsilon)$-approximate solution.
Together, by the result in~\cite{NWF78}, we can show that \alg{WR-Greedy} achieves an approximation factor of $1-e^{-(1-1/e)}-\varepsilon$, as summarized in the following theorem.

\begin{algorithm}[t]
    \caption{{\WRGreedy}: Within-Round Greedy Algorithm} \label{alg:wrgreedy}
    \KwIn{Graph $G=(V,E)$, integers $T$, $k$ and $R$, triggering set distributions.}
    \KwOut{$\cS^o$.}
    
    $\cS^o \leftarrow \emptyset$\;
    \For{$t = 1, 2, \dots, T$}{ 
        \For{$i = 1, 2, \dots, k$}{
            for all $(v,t)\in \cV_t\setminus \cS^o$, 
            estimate ${\rho}(\cS^o \cup \{(v,t)\})$ by simulating the diffusion process $R$ times\; \label{line:simulate}
            $(u,t) \leftarrow \argmax_{(v,t)\in \cV_t\setminus \cS^o} \hat{\rho}(\cS^o \cup \{(v,t)\})$\; 
            $\cS^o \leftarrow \cS^o \cup \{(u,t)\}$\;
        }
    }
    \textbf{return} $\cS^o$.
\end{algorithm}

%
%

\begin{restatable}{theorem}{wrmrim} \label{thm:nonadaptive}
    Let $\cS^*$ be the optimal solution of the non-adaptive MRIM under within-round setting.
    For every $\varepsilon > 0$ and $\ell > 0$, with probability at least
    $1-\frac{1}{n^\ell}$, the output $\cS^o$ of \alg{WR-Greedy}
    satisfies
    \begin{equation*}
    \rho(\cS^o) \geq \left(1 - e^{-(1-\frac{1}{e}) } - \varepsilon\right) \rho(\cS^*),
    \end{equation*}
    if \alg{WR-Greedy}
    uses $R=\lceil 31 k^2n\log(2kn^{\ell+1} T)/\varepsilon^2 \rceil$ as input.
    In this case, the total running time is $O(k^3\ell Tn^2m\log(n T)/\varepsilon^2)$, 
    assuming each simulation finishes in $O(m)$ time.
\end{restatable}
\begin{proof}[Proof (Sketch)]
	The proof follows the same structure as the proof of Theorem 3.7 in~\cite{chen2013information}, but
	it needs to accommodate the new double greedy algorithm structure and the double submodular
	property.
	let $\varepsilon_0 = e^{(1-1/e)} \varepsilon/2$.
	From the inner-submoduarlity property of the MRT model and 
	the proof of Theorem 3.7 in~\cite{chen2013information} based on the Chernoff bound,
	we can conclude that when $R\ge \lceil 27 k^2n\log(2kn^{\ell+1} T)/\varepsilon_0^2 \rceil$,
	for each round $t \in [T]$, with probability at least $1- \frac{1}{n^\ell T}$, 
	the seed set $\cS_t^o$ found by {\WRGreedy} is a $(1-1/e - \varepsilon_0)$ approximation
	of the optimal solution for round $i$ maximizing the marginal gain of
	$\rho(\cS_1 \cup \cS_2 \cup \dots\cup \cS_{t-1} \cup \cS) - \rho(\cS_1 \cup \cS_2 \cup \dots\cup \cS_{t-1}) $.
	Using the union bound, we know that with probability at least $1- \frac{1}{n^\ell}$,
	for all $t \in [T]$ seed set $S_t^o$ is an $(1-1/e - \varepsilon_0)$ approximation
	of the optimal solution for round $t$.
	
	Now at each greedy step, if the new item found is not the one giving the best marginal contribution
	but an $\alpha$ approximation of the optimal marginal solution, an easy extension of
	\cite{NWF78}, already reported in~\cite{GS07}, show that the greedy algorithm can give
	a $1-e^{-\alpha}$ approximate solution to the sumodular maximization problem.
	In our case, consider the outer level when treating each subset $S_i$ as an item, Lemma~\ref{lem:mrtsubmodular} shows that $\rho$ is also submodular in this case 
	(the outer-submodularity), and we just argued in the previous paragraph that
	in each round, the selected $\cS_t^o$ is a $(1-1/e - \varepsilon_0)$ approximation.
	Therefore, the final output $\cS^o$ satisfies
	\begin{equation*}
		\rho(\cS^o) \ge (1-e^{-(1-1/e-\varepsilon_0)}) \rho(\cS^*).
	\end{equation*}
	Because $\varepsilon_0 = e^{(1-1/e)} \varepsilon/2$, it is easy to verify that 
	$1-e^{-(1-1/e-\varepsilon_0)} \ge 1-e^{-(1-1/e)} - \varepsilon$ in this case, and
	it is sufficient to have $R=\lceil 31 k^2n\log(2kn^{\ell+1} T)/\varepsilon^2 \rceil$.
	Finally, the total running time is simply 
	$O(TknRm) = O(k^3\ell Tn^2m\log(n T)/\varepsilon^2)$.
\end{proof}

Compared with Theorem~\ref{thm:crmrim}, the approximation ratio drops from $1/2 -\varepsilon$ to 
	$0.46 - \varepsilon$, but the running time improves by a factor of $T^3$.
One factor of $T$ is because each greedy step of \alg{CR-Greedy} needs to search a space $T$ times larger
	than that of \alg{WR-Greedy}, and the other factor of $T^2$ is because \alg{CR-Greedy} 
	needs more accurate Monte Carlo estimates for each evaluation of $\rho(\cS)$ to avoid deviation,
	again because it searches
	a larger space.
This shows the trade-off between efficiency and approximation ratio, that \alg{CR-Greedy} has a better performance guarantee while \alg{WR-Greedy} is much more efficient.

\section{Adaptive MRIM}\label{sec:AMRIM}

We now study the adaptive MRIM problem. 
Informally, at the beginning of each round, 
one need to determine the seed set for the current round based on 
	the propagation results observed  in previous rounds. 
The formal definition follows the framework and terminology
	provided in~\cite{GoloKrause11} and will be given in Section~\ref{sec:AMRIMDef}.
We then argue about the adaptive submodularity property, 
	propose the adaptive greedy policy and analyze its performance in Section~\ref{sec:adagreedy}.

\subsection{Notations and Definition}\label{sec:AMRIMDef}

We call $(S_t, t)$ an {\em item}, where $S_t$ is the seed set chosen in round $t$.
Let $\cE$ be the set of all the possible items.
For each item $(S_t, t)$,  after the propagation, the nodes and edges participated in the
	propagation are observed as the feedback.
Formally, the feedback is referred to as a {\em state}, which is the subgraph of
	live-edge graph $L_t$ that can be reached by $S_t$.
A {\em realization} is a function $\phi: \cE \rightarrow O$ mapping 
	every possible item $(S_t, t)$ to a state, where $O$ is the set of all possible states.
Realization $\phi$ represents one possible propagation from a possible seed set in a round.
Let $\phi(S_t,t)$ denote the state of $(S_t, t)$ under realization $\phi$. 
We use $\Phi$ to denote a random realization, and the randomness comes from 
	random live-edge graphs $L_1, \ldots, L_T$.
For the adaptive MRIM, in each round $t$, we pick
an item $(S_t, t)$, see its state $\Phi(S_t,t)$, pick the next item $(S_{t+1},t+1)$,
see its state, and so on. 
After each pick, previous observations can be represented as a
{\em partial realization} $\psi$, a function from some subset of $\cE$ to their states.
For notational convenience, we represent $\psi$ as a relation, so that $\psi \subseteq \cE \times \textit{O}$
equals $\{((S_t,t),o)\colon \psi(S_t,t)=o\}$. 
We define $\dom(\psi) = \{ (S_t, t)\colon \exists o, ((S_t, t), o) \in \psi \}$
as the {\em domain} of $\psi$.
A partial realization $\psi$ is {\em consistent with} a realization $\phi$ if they are equal everywhere in the domain of $\psi$,
denoted as $\phi \sim \psi$. 
If $\psi$ and $\psi'$ are both consistent with some $\phi$, and $\dom(\psi) \subseteq \dom(\psi')$,
$\psi$ is a {\em subrealization} of $\psi'$, also denoted as $\psi \subseteq \psi'$.


A {\em policy} $\pi$ is an adaptive strategy for picking items based on partial realizations in $\cE$.
In each round, $\pi$ will pick the next set of seeds $\pi(\psi)$ 
	based on partial realization $\psi$ so far.
If partial realization $\psi$ is not in the domain of $\pi$, 
	the policy stops picking items. 
We use $\dom(\pi)$ to denote the domain of $\pi$. 
Technically, we require that $\dom(\pi)$ be closed under subrealizations: 
If $\psi' \in \dom(\pi)$ and $\psi$ is a subrealization of $\psi'$ then $\psi \in \dom(\pi)$.
We use the notation $E(\pi, \phi)$ to refer to the set of items selected by $\pi$ under realization $\phi$.
The set of items in $E(\pi, \phi)$ is always in the form 
	$\{(S_1, 1), \ldots, (S_t, t)\}$, so sometimes we also refer to it as 
	sequence of seed sets.

We wish to maximize, subject to some constraints, 
a utility function $f\colon 2^\cE \times O^\cE \rightarrow \R_{\geq 0}$ that depends on the picked items and 
the states of them. 
In the adaptive MRIM, $f(\{(S_1, 1), \ldots, (S_t, t)\}, \phi)$ 
	is the total number of active nodes by round $t$ from the respective seed sets,
	i.e., $|\bigcup_{i=1}^t \Gamma(L^\phi_i, S_i)|$ where $L^\phi_t$ is the live-edge graph of round $t$.
Based on the above notations, the expected utility of a policy $\pi$ is
$f_{\avg}(\pi) = \E_\Phi [f(E(\pi, \Phi),\Phi)]$ where the expectation is taken over 
	the randomness of $\Phi$.
Namely, $f_{\avg}(\pi)$ is the expected number of active nodes under policy $\pi$.
Let $\Pi_{T,k}$ be the set of all policies that select seed sets in at most $T$ rounds and
	each seed set has at most $k$ nodes.
The goal of the adaptive MRIM is to find the best policy $\pi$ such that:
$\pi^* = \argmax_{\pi \in \Pi_{T,k}} f_{\avg}(\pi)$.

\subsection{Adaptive Submodularity and Greedy Policy}
\label{sec:adagreedy}

Given a partial realization $\psi$ of $t-1$ rounds with 
$\dom(\psi) = \{(S_1, 1), (S_2, 2), \dots, (S_{t-1}, t-1)\}$, 
and the seed set $S_{t}$ for round $t$, 
the {\em conditional expected marginal benefit} of item $(S_{t},t)$ conditioned on having observed $\psi$ 
    is defined as
    \begin{align*}
        \Delta((S_{t},t) | \psi) 
         = \E_\Phi \left[ f(\dom(\psi) \cup \{(S_{t},t)\}, \Phi) - f(\dom(\psi), \Phi) \mid \Phi \sim \psi \right]. 
    \end{align*}
The {\em conditional expected marginal gain} of a policy $\pi$ is defined as
\begin{align*}
    \Delta(\pi | \psi) 
    = \E \left[ f(\dom(\psi) \cup E(\pi, \Phi), \Phi) - f(\dom(\psi), \Phi) \mid \Phi \sim \psi \right].
\end{align*}
The adaptive MRIM satisfies the adaptive monotonicity and submodularity shown as below.
The proofs require a careful analysis of the partial realization in the MRT model and is given in \OnlyInShort{\cite{sun2018multi}.} \OnlyInFull{the appendix.}

\begin{restatable}{lemma}{adamonotone}[Adaptive Monotonicity] 
    For all $t>0$, for all partial realization $\psi$ with $t-1$ rounds and $\Pr[\Phi \sim \psi]>0$, and for all item $(S_t, t)$, we have:
    \begin{equation*}
    \Delta((S_t,t) | \psi) \geq 0.
    \end{equation*}
\end{restatable}

\begin{restatable}{lemma}{adasubmodular}[Adaptive Submodularity] \label{lem:adaSubmodular}
    For all $t>0$, for all partial realization $\psi$ with $i-1$ rounds and partial realization $\psi'$ such that $\psi'$ is a subrealization of $\psi$, i.e., $\psi' \subseteq \psi$,
    and for all item $(S_{i},i)$, we have
    \begin{equation*}
    \Delta((S_{t},t) | \psi') \geq  \Delta((S_{t},t) | \psi).
    \end{equation*}
\end{restatable}



Following the framework of~\cite{GoloKrause11}, adaptive monotonicity and adaptive submodularity
	enable an adaptive greedy policy with a constant approximation of the optimal
	adaptive policy.
{\alg{AdaGreedy} (Algorithm~\ref{alg:adaptivegreedy}) is the greedy adaptive policy for MRIM.
Note that adaptive algorithms operate at per round base --- it takes feedback from
	previous rounds and selects the item for the current round, and then obtain new feedback.
Thus we present \alg{AdaGreedy} for a generic round $t$.
Besides the problem input such as the graph $G$, the triggering model parameters, 
	seed set budget $k$, and simulation number $R$, \alg{AdaGreedy} takes the set of already activated nodes $A_{t-1}$ 
	as the feedback from the previous rounds, and aims at finding the seed set $S_t$
	of size $k$ to maximize the expected marginal gain $\Delta((S_{t},t) | \psi)$, which
	is the expected number of newly activated nodes in round $t$.
However, this problem is NP-hard, so we use a Monte Carlo greedy approximation
	$\alg{MC-Greedy}$ algorithm to find an approximate solution. 
$\alg{MC-Greedy}$ greedily finds the seed with the maximum estimated marginal influence spread until $k$ seeds being selected, where the marginal influence spread of adding an unselected seed is estimated by simulating the propagation $R$ times.
In \alg{AdaGreedy}, $\alg{MC-Greedy}$ won't count the influence of a node if it has been activated in previous rounds. 
The rationale is that maximizing the expected marginal gain $\Delta((S_{t},t) | \psi)$
	is equivalent to the weighted influence maximization task in which we treat nodes in $A_{t-1}$ with
	weight $0$ and other nodes with weight $1$, and we maximize the expected total
	weight of the influenced nodes.
By \cite{mossel2007}, we know that the weighted version is also monotone and submodular, 
so we could use a greedy algorithm to obtain a seed set $S_t$
as a $(1-1/e -\varepsilon)$ approximation of the best seed set for round $t$.
We could use $R$ Monte Carlo simulations to estimate the weighted influence spread,
and $R$ is determined by the desired approximation accuracy $\varepsilon$.
Once the seed set $S_t$ is selected for round $t$, the actual propagation from $S_t$
	will be observed, and the active node set $A_t$ will be updated as the feedback for
	the next round.

The following theorem summarizes the correctness and the time complexity of \alg{AdaGreedy}.

\begin{algorithm}[t] 
	\caption{\alg{AdaGreedy}: Adaptive Greedy for Round $t$} \label{alg:adaptivegreedy}
	\KwIn{Graph $G=(V,E)$, integers $T$, $k$ and $R$, triggering set distributions, 
		active node set $A_{t-1}$ by round $t-1$.}
	\KwOut{Seed set $S_t$, updated active nodes $A_t$}
	
	$S_t\leftarrow $ \alg{MC-Greedy($G, A_{t-1}, k, R$)}\tcp*{Monte Carlo Greedy} 
	Observe the propagation of $S_t$, update activated nodes $A_t$\;
	\textbf{return} $(S_t,t)$, $A_t$.
\end{algorithm}

%
%
%
%

\begin{restatable}{theorem}{adagreedy} \label{thm:AdaGreedy}
	Let $\pi^{\ag}$ represents the policy corresponding to the \alg{AdaGreedy} algorithm. 
	For any $\varepsilon > 0$ and $\ell > 0$, 
		if we use $R=\lceil 31 k^2n\log(2kn^{\ell+1} T)/\varepsilon^2 \rceil$ simulations
		for each influence spread estimation, then
		with probability at least $1-\frac{1}{n^\ell}$, 
	\begin{equation*}
f_{\avg}(\pi^{\ag}) \geq \left(1 - e^{-(1-\frac{1}{e})}-\varepsilon \right) f_{\avg}(\pi^*).
\end{equation*}
 In this case, the total running time for $T$-round \alg{AdaGreedy}
 	is $O(k^3\ell Tn^2m\log(n T)/\varepsilon^2)$.
\end{restatable}

%
%
%
%


\section{Scalable Implementations}\label{sec:RRset}

In this section, we aim to speed up \alg{CR-Greedy}, \alg{WR-Greedy} and \alg{AdaGreedy} by the reverse influence sampling~\cite{BorgsBrautbarChayesLucier,tang14,tang15}.

A {\em Reverse-Reachable (RR) set $R_v$ rooted at node $v\in V$}
is 
the set of nodes which are reached by reverse simulating a propagation from $v$ in the triggering model.
Equivalently, 
$R_v$ is the set of nodes in a random live-edge graph which can reach $v$.
We use $\rroot(R_v)$ to denote its root $v$.
We define a {\em (random)  RR set} $R$ is a RR set rooted at a node picked uniformly at random from $V$,
then for any seed set $S\subseteq V$, its influence spread
\begin{equation}
\sigma(S) = n \cdot \E[\I\{ S\cap R \ne \emptyset \}], \label{eq:RRset}
\end{equation}
where $n=|V|$, $\I\{\}$ is the indicator
function, and the expectation is taken over the randomness of $R$: randomness of root node and randomness of live-edge graph.
The property implies that we can accurately estimate 
    the influence spread of any possible seed set $S$ by sampling enough RR sets.
More importantly, 
by Eq.~\eqref{eq:RRset} the optimal seed set
can be found by seeking the optimal set of nodes that intersect with (a.k.a. cover) 
the most number of
RR sets, which is a max-cover problem.
Therefore, a series of near-linear-time algorithms are
developed~\cite{BorgsBrautbarChayesLucier,tang14,tang15} based on the above observation.
All RR-set algorithms have the same structure of two phases.
In Phase 1, the number of RR sets needed is estimated, and in Phase 2, these RR sets are
generated and greedy algorithm is used on these RR sets to find the $k$ nodes
that cover the most number of RR sets.
All algorithms have the same Phase 2, but Phase 1 is being improved from one to another
so that less and less RR sets are needed.
Our algorithms are based on \alg{IMM} proposed in \cite{tang15}.

\begin{algorithm}[t] 
	\caption{{\CRNS}: Cross-Round Node Selection} \label{alg:crnr}
	\KwIn{Multi round RR vector sets $\cM$, $T$, $k$}
	\KwOut{seed sets $\cS^o$}
	
	Build count array: $c[(u, t)] = \sum_{(u,t) \in \cM}|(u, t)|$, $\forall (u, t) \in \cV$\; \label{line:crgen1}
	Build RR set link: $RR[(u, t)]$, $\forall (u, t) \in \cV$\; \label{line:crgen2}
	For all $\cR \in \cM$, $covered[\cR] = false$\;
	$\cS^o \leftarrow \emptyset$; $\cC\leftarrow \cV$;
	$c_1, c_2, \dots, c_T \leftarrow 0$\; 
	\For{$i = 1 \; to \; Tk$ \label{line:crfor1}}{
		$(u, t) \leftarrow \argmax_{(u', t')\in \cC\setminus \cS^o} c[(u', t')]$\;\label{line:crf2}
		$\cS^o \leftarrow \cS^o \cup \{(u, t)\}$;  $c_{t} \leftarrow c_{t} + 1$\; 
		\If{$c_{t} == k$}{ $\cC = \cC \setminus \{(v,t) \mid v\in V\}$}\label{line:crif1}
		\For{all $\cR \in RR[(u, t)] \wedge covered[\cR]= false$\label{line:crfor2}}{
			$covered[\cR]= true$\;
			for all $(u', t') \in \cR \wedge (u', t') \ne (u, t)$ to do $c[(u', t')] = c[(u', t')] -1$\;\label{line:crfor3}
		}
	}
	\bf{return} $\cS^o $.
\end{algorithm}

\begin{algorithm}[t] 
	\caption{{\CRIMM}: Non-adaptive IMM Algorithm for Cross-Round
	} \label{alg:crimm}
	\KwIn{Graph $G=(V,E)$, round number $T$, budget $k$, accuracy parameters $(\varepsilon, \ell)$, triggering set distributions}
	\KwOut{seed set $\cS$}
	
	\tcp{Phase 1: Estimating the number of multi-round RR sets needed, $\theta$}
	$\newalg{\ell \leftarrow \ell + \ln 2 / \ln{n}} $;
	$\newalg{\cM \leftarrow \emptyset}$;  $LB \leftarrow 1$; $\varepsilon' \leftarrow \sqrt{2}\varepsilon$\;
	$\alpha \leftarrow \sqrt{\ell \ln{n} + \ln{2}}$;
	$\beta \leftarrow \sqrt{(1-1/2)\cdot (T\ln{\binom{n}{k}}+\alpha^2)}$\;
	$\lambda' \leftarrow [(2+\frac{2}{3}\varepsilon') \cdot (T\ln{{\binom{n}{k}}}+\ell \cdot \ln{n}+\ln{\log_2{n}})\cdot n]/\varepsilon'^2$\;
	$\lambda^* \leftarrow 2 nT \cdot ((1-1/e) \cdot \alpha + \beta)^2 \cdot \varepsilon^{-2}$\;
	
	\For{$i = 1 \; to \; \log_2{(\newalg{n} - 1)}$ \label{line:cfor1}}{
		$x \leftarrow \newalg{n}/2^i$\; \label{line:cassignx}
		$\theta_i \leftarrow \lambda'/x_i$\; \label{line:cgenRRset1b} 
		\While{$\newalg{|\cM|}< \theta_i$}{
			Select a node $u$ from $V$ uniformly at random\; \label{line:csample1}
			Generate RR-vector $\newalg{\cR}$ from $u$, and insert it into $\newalg{\cM}$\;  \label{line:cgsample1}
		}\label{line:cgenRRset1e} 
		\newalg{$\cS_i \leftarrow \CRNS(\cM,T, k)$}\; \label{line:cnodeselect1}
		\If{$\newalg{n} \cdot \newalg{F_{\cM}(\cS_i)} \geq (1 + \varepsilon') \cdot x$}{
			\label{line:cestimate1}
			$LB \leftarrow \newalg{n} \cdot \newalg{F_{\cM}(\cS_i)}/(1 + \varepsilon')$\; 
			\label{line:cestimate2}
			break\;
		} 
	}
	$\theta \leftarrow \lambda^* / LB$\;
	\While{$\newalg{|\cM|} \leq \theta$}{ \label{line:ccheckLB}
		Select a node $u$ from $V$  uniformly at random\;\label{line:csample2}
		Generate \newalg{$\cR$} for $u$, and insert it into \newalg{$\cM$}\; \label{line:cgsample2}		
	}
	\tcp{Phase 2: Generate $\theta$ RR-vector sets and select seed nodes}
	\newalg{$\cS \leftarrow \CRNS(\cM,T, k)$}\;
	
	\bf{return} $\cS$.
\end{algorithm}

\subsection{Non-Adaptive IMMs}
For the non-adaptive MRIM, we define the {\em multi-round reverse-reachable (RR) set $\cR_v$ rooted at node $v$} for the MRT model as
$\cR_v:=\bigcup_{t=1}^T\cR_{v,t}$ where $\cR_{v,t}$ denotes a RR set rooted at $v$ of round $t$ in pair notation. 
$\cR_v$ is generated by independently reverse simulating the propagation $T$ rounds from $v$ and then aggregating them together. 
Let $\rroot(\cR_v):=v$.
A (random) multi-round RR set $\cR$ is a multi-round RR set rooted at a node picked uniformly at random from $V$.
We use $\cM$ to denote the set of $\cR_v$.
We are now ready to explain the cross-round and within-round non-adaptive \alg{IMM}.

\subsubsection{Cross-Round Non-Adaptive IMM}\label{sec:CRIMM}

In cross-round setting, if $\cR$ is a random multi-round RR set,
then for any seed set $\cS$, its influence spread satisfies the following lemma.

\begin{restatable}{lemma}{crimm} \label{lem:crimm}
	For any node-round pair seed set $\cS$,
	\begin{equation*}
	\rho(\cS) = n  \cdot \E[\I\{ \cS\cap \cR \ne \emptyset \}], 
	\end{equation*}
	where the expectation is taken over the randomness of $\cR$.
\end{restatable}

The Lemma \ref{lem:crimm} implies we can sample enough multi-round RR sets to accurately estimate the influence spread of $\cS$.

\alg{\CRIMM} is very similar to standard \alg{IMM} and only has a few differences include several points.
First, \alg{\CRIMM} generates multi-round RR sets $\cR$ from roots in $V$ (lines~\ref{line:csample1} and~\ref{line:csample2}). 
Second, we need to adjust $\ell$ to be $\ell + \log(2)/\log n$, and $\varepsilon$.
This is to guarantee that in each round we have probability at least $1 - 1/(n^\ell T)$ to
have $\cS_t$ as a $(\frac{1}{2} - \varepsilon)$ approximation, so that the result for the
total $T$ rounds would come out correctly as stated in the following theorem.
Third, we use new $F_{\cM}(\cdot)$ denotes the fraction of multi-round RR sets in $\cR$ that are covered by a node set $\cS$ in algorithm $\CRIMM$ (lines \ref{line:cestimate1} and \ref{line:cestimate2}).
Forth, \alg{\CRNS} returns $Tk$ seeds from the total $T$ rounds (line \ref{line:crfor1}).
Last, if the budget for some round $t$ exhausts, then the remaining nodes of $\cV_t$ are removed from $\cC$ in \alg{\CRNS} (line \ref{line:crif1}).

By an analysis similar to that of the \alg{IMM} algorithm~\cite{tang15}, we can show that
	our \alg{\CRIMM} achieves $1/2-\varepsilon$ approximation with expected running time
	$O(T^2(k+\ell)(m+n)\log (n) / \varepsilon^2)$.

\begin{restatable}{theorem}{crmrimt} \label{thm:crmrimt}
	Let $\cS^*$ be the optimal solution of the non-adaptive MRIM.
	For every $\varepsilon > 0$ and $\ell > 0$, with probability at least
	$1-\frac{1}{n^\ell}$, the output $\cS^o$ of the cross-round algorithm  \alg{\CRIMM}
	satisfies
	\begin{equation*}
		\rho(\cS^o) \geq \left(\frac{1}{2} - \varepsilon\right) \rho(\cS^*),
	\end{equation*}
	In this case, the total running time for $T$-round \alg{\CRIMM}
	is $O(T^2(k + \ell)(n + m)\log{(n)}/\varepsilon^2)$.
\end{restatable}

\subsubsection{Within-Round Non-Adaptive IMM}\label{sec:WRIMM}

In the within-round setting, the idea is to use the IMM algorithm in each round to select $k$ seeds.
However, seeds selected in earlier rounds may already influence some nodes, so when we select roots
	for a later round $t$ and generate RR sets for round $t$, the roots should not be selected uniformly
	at random.
Instead, we want to utilize the idea derived from the following lemma.
Let $\cS^{t-1} :=\bigcup_{t'=1}^{t-1} \cS_{t'}$ be the set of seed pairs in first $t-1$ rounds, and let $\cS_t$ be a set of seed pairs in round $t$.
Similarly, let $\cR_v^{t-1}:=\bigcup_{t'=1}^{t-1}\cR_{v,t'}$ be the set of RR sets (in pairs) in first $t-1$ rounds, 
	$\cR_t$ be the RR set (in pairs) for round $t$.
The marginal influence spread of $\cS_t$ in round $t$ is 
%
%
 \begin{restatable}{lemma}{wrimm} \label{lem:wrimm}
 	For any node-round pair seed set $\cS$,
 	\begin{align*}
 	& \rho(\cS^{t-1}\cup \cS_t) - \rho(\cS^{t-1}) \\
 	& = n  \cdot \E[\I\{ (\cS^{t-1} \cap \cR^{t-1} = \emptyset) \wedge (\cS_t \cap \cR_{t} \ne \emptyset)\}] \\
 	& = n \cdot \Pr\{\cS^{t-1} \cap \cR^{t-1} = \emptyset \} \cdot 
 	\E[\I\{\cS_t \cap \cR_{t} \ne \emptyset\} \mid \cS^{t-1} \cap \cR^{t-1} = \emptyset],
 	\end{align*}
 	where the expectation is taken over the randomness of $\cR$.
 \end{restatable}     
 
 The above lemma suggests that when we want to generate an RR set for round $t$, we should also generate
 RR sets for earlier rounds and check that if any of them is intersecting with the seed set in the same 
 round, and if so, the RR set for round $t$ is invalid and we need to regenerate an RR set again.
 By following this, the implementation is similar to the $\CRIMM$, and the only difference between them is that the within-round non-adaptive algorithm $\WRIMM$ selects $k$ seeds round by round.
The resulting algorithm would have the approximation guarantee
	of $1 - e^{-(1-\frac{1}{e})}-\varepsilon$, but 
	it does not have significant running time improvement over $\CRIMM$, 
	since it wastes many RR set generations.
 We find a better heuristic to use the roots generated for the previous rounds.
 In particular, for each $t\ge 2$, after we finished selecting $k$ seeds in round $t-1$, some RR sets
 	are removed since they are covered by seeds selected. 
 The remaining roots are exactly the ones whose RR sets do not intersect with seed sets in the first
 	$t-1$ rounds.
 Therefore, their distribution is close to the distribution of the valid roots satisfying 
 	$\cS^{t-1} \cap \cR^{t-1} = \emptyset$ in Lemma~\ref{lem:wrimm}.
 Hence, in round $t$, we sample roots using the remaining roots from round $t-1$.
 This gives a close estimate of the marginal influence spread.
Due to the complicated stochastic dependency of RR sets from round to round, the exact theoretical analysis
	of this improvement is beyond our reach, and thus we propose it as an efficient heuristic.

The resulting algorithm \alg{\WRIMM}  runs almost exactly like running a copy of the standard \alg{IMM} 
	for each round $t$.
The only difference is that in the standard \alg{IMM}, the root of an RR set is always sampled uniformly at random
	from all nodes, but in \alg{\WRIMM} the root of an RR set in round $t$ is sampled uniformly at random
	from the remaining roots left in round $t-1$.
Moreover, for each round $t$, we set the per-round approximation error bound
	$\varepsilon_0 = e^{(1-1/e)} \varepsilon/2$
	and replace $\ell $ with $\ell + \log(2T)/\log n$.
This is consistent with the adaptive IMM setting in Section~\ref{sec:AdaIMM}, and
	with the setting used in the proof of Theorem~\ref{thm:nonadaptive} for the within-round
	greedy algorithm.
%
 The details of \alg{\WRIMM} are provided in \OnlyInShort{\cite{sun2018multi}.} \OnlyInFull{the Appendix~\ref{app:WRIMM}.}

\subsection{Adaptive IMM} \label{sec:AdaIMM}
We use $M$ to denote the set of $R_v$.
In the adaptive setting, as already mentioned in Section~\ref{sec:adagreedy}, 
some nodes $A_{t-1}$ are already activated, so they won't
contribute to the influence spread in round $t$.
Therefore, we are working on the weighted influence maximization problem, where nodes
in $A_{t-1}$ has weight $0$ and nodes in $V\setminus A_{t-1}$ has weight $1$.
Let $\sigma^{-A_{t-1}}(S)$ be the weighted influence spread according to the above weight.
Let $R^{-A_{t-1}}$ be a random RR set where the root $v$ is selected from $V\setminus A_{t-1}$
uniformly at random, and then reverse simulate from $v$ to get the RR set.
Then we obtain the result similar to Eq.~\eqref{eq:RRset}:

\begin{restatable}{lemma}{adaimm} \label{lem:ARRset}
	For any seed set $S\subseteq V$, 
	\begin{equation}
	\sigma^{-A_{t-1}}(S) = (n-|A_{t-1}|) \cdot \E[\I\{ S\cap R^{-A_{t-1}} \ne \emptyset \}], \label{eq:AdaRRset}
	\end{equation}
	where the expectation is taken over the randomness of $R^{-A_{t-1}}$.
\end{restatable}
\begin{proof}
	\begin{align}
		& \E[\I\{ S\cap R^{-A_{i-1}} \ne \emptyset \}]  \nonumber \\
		& = \sum_{v\in V\setminus A_{i-1}} \Pr\{v =\rroot(R^{-A_{i-1}}) \} \cdot \nonumber \\
		& \quad \quad \E[\I\{ S\cap R^{-A_{i-1}} \ne \emptyset \} \mid v =\rroot(R^{-A_{i-1}})] 
		\nonumber \\
		& = \frac{1}{n-|A_{i-1}|} \sum_{v\in V\setminus A_{i-1}} 
		\E[\I\{ S\cap R_v^{-A_{i-1}} \ne \emptyset \} ] \nonumber  \\
		& = \frac{1}{n-|A_{i-1}|} \sum_{v\in V\setminus A_{i-1}} 
		\E[\I\{ v\in \Gamma(L, S)\} ] \label{eq:RRset2L} \allowdisplaybreaks \\
		& = \frac{1}{n-|A_{i-1}|} \E \left[ \sum_{v\in V\setminus A_{i-1}} \I\{ v\in \Gamma(L, S)\} \right] \nonumber \\
		& = \frac{1}{n-|A_{i-1}|} \E[|\Gamma(L, S) \setminus A_{i-1} | ] \nonumber \\
		& =  \frac{1}{n-|A_{i-1}|} \sigma^{-A_{i-1}}(S), \nonumber 
	\end{align}
	where Eq.~\eqref{eq:RRset2L} is based on the equivalence between RR sets and live-edge graphs,
	and the expectation from this point on is taken over the random live-edge graphs $L$. 
\end{proof}

Therefore, with Eq.~\eqref{eq:AdaRRset}, the same RR-set based algorithm can be used,
and we only need to properly change the RR-set generation process and the
estimation process.
\alg{AdaIMM} is based on the \alg{IMM} algorithm in~\cite{tang15}.
The main differences from the standard \alg{IMM} include several points.
First, whenever we generate new RR sets in round $t$, we only start from roots in
$V \setminus A_{t-1}$ as explained by Lemma~\ref{lem:ARRset}.
Second, when we estimate the influence spread, we need to adjust it using
$n_a = n - |A_{t-1}|$ again by Lemma~\ref{lem:ARRset}.
Third, we need to adjust $\ell$ to be $\ell + \log(2T)/\log n$,
and $\varepsilon$ to $\varepsilon_0 = e^{(1-1/e)} \varepsilon/2$.
This is to guarantee that in each round we have probability at least $1 - 1/(n^\ell T)$ to
have $S_t$ as a $(1-1/e - \varepsilon_0)$ approximation, so that the result for the
total $T$ rounds would come out correctly as stated in the following theorem.
The details are shown in \OnlyInShort{\cite{sun2018multi}.} \OnlyInFull{the appendix.}

\begin{restatable}{theorem}{adaimmt} \label{thm:AdaIMM}
	Let $\pi^{\ai}$ represents the policy corresponding to the \alg{AdaIMM} algorithm. 
	For any $\varepsilon > 0$ and $\ell > 0$, with probability at least $1-\frac{1}{n^\ell}$, 
	\begin{equation*}
	f_{\avg}(\pi^{\ai}) \geq \left(1 - e^{-(1-\frac{1}{e})}-\varepsilon \right) f_{\avg}(\pi^*).
	\end{equation*}
	In this case, the total running time for $T$-round \alg{AdaIMM}
	is $O(T(k + \ell)(n + m)\log{(nT)}/\varepsilon^2)$.
\end{restatable}
\begin{proof}[Proof (Sketch)]
	Let $\varepsilon_0 = e^{(1-1/e)} \varepsilon/2$, 
	$\ell' = \ell + \log (2T) / \log n $.
	As explained already, one round \alg{AdaIMM} is essentially the same as \alg{IMM}
	with parameters $\varepsilon_0$ and $\ell'$.
	Thus, following the result in~\cite{tang15}, we know that for each round $i$,
	with probability at least $ 1 - 2/n^{\ell'} = 1 - 1/(n^\ell T)$, 
	output $S_i$ is a $(1 - 1/e - \varepsilon_0)$ approximation of the
	best seed set for this round.
	Then following the similar arguments as in Theorems~\ref{thm:nonadaptive} and~\ref{thm:AdaGreedy},
	we know that across all $T$ rounds, 
	with probability at least $1 - 1/n^\ell$,
	\begin{align*}
		f_{\avg}(\pi^{\ai}) &  \geq \left(1 - e^{-(1-\frac{1}{e} - \varepsilon_0)}\right) f_{\avg}(\pi^*)\\
		& \ge \left(1 - e^{-(1-\frac{1}{e})}-\varepsilon \right) f_{\avg}(\pi^*).
	\end{align*}
	Thus, the theorem holds.
\end{proof}

Theorem~\ref{thm:AdaIMM} clearly shows that the \alg{AdaIMM} algorithm is near linear time,
and its theoretical time complexity bound is much better than
the one in Theorem~\ref{thm:AdaGreedy} for the \alg{AdaGreedy} algorithm.

In practice, we can further improve the \alg{AdaIMM} algorithm by incremental computation.
In particular, RR sets generated in the previous $t-1$ rounds can be used for round $t$,
and we only need to remove those rooted at nodes in $A_{t-1}$.
Searching the lower bound $LB$ can also be made faster by utilizing the $x_{t-1}$ already
obtained in the previous round.
Our experiments would use such improvements.

\section{Experimental Evaluation} \label{sec:experiment}

We conduct experiments on two real-world social networks to test the performance 
	of our algorithms.
We use the independent cascade model for the tests.
The influence probabilities on the edges are learned from the real-world trace data
	in one dataset and synthetic in the other dataset, as explained below.


\subsection{Data Description}

\textbf{Flixster}.
The Flixster dataset is a network of American social movie discovery service (www.flixster.com). 
To transform the dataset into a weighted graph, 
each user is represented by a node, 
and a directed edge from node $u$ to $v$ is formed if $v$ rates one movie shortly 
after $u$ does so on the same movie. 
The dataset is analyzed in \cite{barbieri2012topic}, 
and the influence probability are learned by the topic-aware model. 
We use the learning result of \cite{barbieri2012topic} in our experiment, 
which is a graph containing \num{29357} nodes and \num{212614} directed edges. 

\noindent\textbf{NetHEPT}. 
The NetHEPT dataset \cite{chen08} is extensively used in many influence maximization studies. 
It is an academic collaboration network from the ``High Energy Physics Theory'' section of arXiv from 1991 to 2003, 
where nodes represent the authors and each edge represents one paper co-authored by two nodes. 
There are \num{15233} nodes and \num{58891} undirected edges (including duplicated edges) in the NetHEPT dataset.
We clean the dataset by removing those duplicated edges and obtain a directed graph $G = (V,E)$,
$|V|$ = \num{15233}, $|E|$ = \num{62774} (directed edges). 

\subsection{Result}

We test all six algorithms proposed in the experiment, for the MRIM task of $T$ rounds with $k$ seeds in each round.
We use $R=10000$ for the Monte Carlo simulation of the influence spread of each candidate seed set for all non-adaptive algorithms.
The lazy evaluation technique of~\cite{Leskovec07} to optimize the greedy selection is applied to \alg{\CRGreedy}, \alg{\WRGreedy} and \alg{AdaGreedy}.
Besides the six proposed algorithms, we also propose two baseline non-adaptive algorithms (\alg{SG} and \alg{SG-R}) using the classical single-round algorithms directly for the multi-round influence maximization problem.  \alg{SG} simply selects $Tk$ seed nodes using the single-round greedy algorithm, and then allocates the first $k$ seeds as $S_1$ for the first round, second $k$ seeds as $S_2$ for the second round, and so on. 
\alg{SG-R} only selects $k$ seeds greedily, and then reuse the same $k$ seeds for all $T$ rounds.

In the tests, we set $T = 5$ and $k=10$ as the default, and we also 
test different combinations of $T$ and $k$ while keeping $Tk$ to be the same,
to see the effect of different degrees of adaptiveness. 

\subsubsection{Influence Spread Performance}
We test the performance on the influence spread for all algorithms 
	introduced in the above section.
For non-adaptive algorithms, 
	for each selected seed set sequence, 
	we do $\num{10000}$ forward simulations and take the average to obtain its estimated
	influence spread.
For adaptive algorithms, to obtain their expected influence spread over multiple
	real-world propagation simulations, we have to re-run the algorithm in each round
	after obtaining the feedback from the previous rounds.
Thus, it would be too time consuming to also run $10000$ adaptive simulations
	for the adaptive algorithms.
In stead, for NetHEPT we use $150$ simulations and for Flixster we use
	$100$ simulations.
To make fair comparisons, we include confidence intervals in the obtained 
	influence spread results, so that the number of simulations used for the estimation
	is taken into the consideration.

Tables \ref{table::hep_per}  and \ref{table::flixster_per} show the influence
	spread results for NetHEPT and Flixster datasets.
All five round results are shown, one for each column.
Each row is for one algorithm, and the number in the first line of the row
	is the empirical average of the influence spread, and the line below is
	the $95\%$ confidence interval.
Parameter $R$ records the number of simulations used to obtain the average spread.

\begin{table}[t]
	\centering
	\caption{The performance of influence spread on NetHEPT.}
	\resizebox{3.45in}{!}{%
	\label{table::hep_per}
	\begin{tabular}{|c|c|c|c|c|c|}
		\hline
		\multirow{2}{*}{Method/Simulations} & \multicolumn{5}{c|}{Round}                                                                                   \\ \cline{2-6} 
		& 1                  & 2                    & 3                  & 4                  & 5                      \\ \hline
		\alg{SG}                                 & 290.1              & 505.7                & 688.6              & 868.2              & 1027.3                 \\ 
		(R = 10000)                           & {[}288.8, 291.4{]} & {[}504.0, 507.3{]}   & {[}686.6, 690.4{]} & {[}866.2, 870.2{]} & {[}1025.2, 1029.4{]}   \\ \hline
		\alg{SG-R}                               & 289.5              & 516.3                & 714.0              & 884.9              & 1042.0                 \\ 
		(R = 10000)                           & {[}288.2, 290.8{]} & {[}514.6, 518.0{]}   & {[}712.0, 716.0{]} & {[}882.7, 887.1{]} & {[}1039.7, 1044.2{]}   \\ \hline
		\alg{E-WR-Greedy}                            & 290.7              & 528.9                & 738.8              & 930.2              & 1097.6.9                 \\ 
		(R = 10000)                           & {[}289.4, 292.0{]} & {[}527.2, 530.6{]} & {[}736.9, 740.8{]} & {[}928.0, 932.3{]} & {[}1095.3, 1099.8{]}   \\ \hline
		\alg{WR-IMM}                            & 290.9              & 532.8                & 745.3              & 930.1              & 1093.1                 \\ 
		(R = 10000)                           & {[}289.7, 292.3{]} & {[}531.1, 534.5{]} & {[}743.2, 747.3{]} & {[}928.0, 932.2{]} & {[}1090.8, 1095.3{]}   \\ \hline
		\alg{CR-Greedy}                            & 267.8              & 528.7                & 730.4              & 938.5              & 1121.3                 \\ 
		(R = 10000)                           & {[}266.5, 269.1{]} & {[}527.2, 530.4{]} & {[}728.5, 732.4{]} & {[}933.7, 937.8{]} & {[}1119.0, 1123.5{]}   \\ \hline
		\alg{CR-IMM}                            & 283.0              & 517.4                & 721.9              & 931.6              & 1129.7                 \\ 
		(R = 10000)                           & {[}281.7, 284.2{]} & {[}515.7, 519.2{]} & {[}720.0, 723.9{]} & {[}929.4, 933.7{]} & {[}1127.7, 1131.9{]}   \\ \hline
		\alg{AdaGreedy}                             & 288.3              & 533.4                & 758.1              & 960.1              & 1141.5                 \\ 
		(R = 150)                             & {[}276.7, 299.7{]} & {[}519.4, 547.3{]}   & {[}743.6, 772.7{]} & {[}943.9, 976.3{]} & {[}1123.7, 1160.0{]}   \\ \hline
		\alg{AdaIMM}                             & \textbf{291.8}              & \textbf{544.4}                & \textbf{761.8}              & \textbf{965.8}              & \textbf{1146.3}                 \\ 
		(R = 150)                             & {[}281.3, 302.4{]} & {[}531.6, 557.2{]}   & {[}746.6, 776.9{]} & {[}949.7, 982.0{]} & {[}1129.1, 1163.5{]} \\ \hline
	\end{tabular}
}
\end{table}

\begin{table}[t]
	\centering
	\caption{The performance of influence Spread on Flixster.}
	\resizebox{3.45in}{!}{%
		\label{table::flixster_per}
		\begin{tabular}{|c|c|c|c|c|c|}
			\hline
			\multirow{2}{*}{Method/Simulations} & \multicolumn{5}{c|}{Round}                                                                                   \\ \cline{2-6} 
			& 1                  & 2                  & 3                    & 4                    & 5                    \\ \hline
			\alg{SG}                                 & 558.8              & 936.2              & 1200.3               & 1437.9               & 1631.5               \\ 
			(R = 10000)                           & {[}557.3, 560.3{]} & {[}934.5, 937.9{]} & {[}1198.4, 1202.2{]} & {[}1435.9, 1439.9{]} & {[}1629.5, 1633.6{]} \\ \hline
			\alg{SG-R}                               & \textbf{559.8}              & 949.2              & 1262.6               & 1530.3               & 1764.9               \\ 
			(R = 10000)                           & {[}558.3, 561.3{]} & {[}947.4, 951.0{]} & {[}1260.6, 1264.5{]} & {[}1528.2, 1532.4{]} & {[}1762.7, 1767.0{]} \\ \hline
			\alg{E-WR-Greedy}                             & 557.8              & 976.5              & 1304.2               & 1587.8               & 1840.0               \\ 
			(R = 10000)                           & {[}556,3 559.2{]}  & {[}974.8, 978,3{]} & {[}1302.2, 1306.1{]} & {[}1585.8, 1580.8{]} & {[}1838.0, 1842.1{]} \\ \hline
			\alg{WR-IMM}                            & 558.1              & 967.5                & 1306.9              & 1599.1              & 1836.4                 \\ 
			(R = 10000)                           & {[}556.7, 559.6{]} & {[}965.7, 969.3{]} & {[}1306.9, 1308.9{]} & {[}1597.1, 1601.1{]} & {[}1834.3, 1838.5{]}   \\ \hline
			\alg{CR-Greedy}                            & 519.9              & 948.6                & 1295.7              & 1593.5              & 1863.8                 \\ 
			(R = 10000)                           & {[}518.4, 521.5{]} & {[}946.7, 950.5{]} & {[}1293.7, 1297.7{]} & {[}1591.4, 1595.5{]} & {[}1861.7, 1865.9{]}   \\ \hline
			\alg{CR-IMM}                            & 521.7              & 935.8                & 1275.3              & 1585.9              & 1865.1                 \\ 
			(R = 10000)                           & {[}521.7, 523.2{]} & {[}933.1, 937.0{]} & {[}1273.3, 1277.3{]} & {[}1583.8, 1588.0{]} & {[}1863.1, 1867.3{]}   \\ \hline
			\alg{AdaGreedy}                             & 557.8              & 977.8              & 1307.7               & 1605.2               & 1861.8               \\ 
			(R = 100)                             & {[}539.8, 580.5{]} & {[}956.2, 999.1{]} & {[}1291.1, 1324.3{]} & {[}1588.1, 1622.3{]} & {[}1845.3, 1878.3{]} \\ \hline
			\alg{AdaIMM}                             & 555.5              & \textbf{977.9}              & \textbf{1317.2}               & \textbf{1613.2}               & \textbf{1872.5}               \\ 
			(R = 100)                             & {[}542.3, 568.6{]} & {[}962.9, 993.0{]} & {[}1300.8, 1333.5{]} & {[}1594.2, 1632.1{]} & {[}1853.0, 1891.9{]} \\ \hline
		\end{tabular}
	}
\end{table}

Several observations can be made from these results.
First, all six proposed algorithms in this paper performs significantly better than
	the baseline algorithms \alg{SG} and \alg{SG-R}.
Besides the cross-round non-adaptive algorithms, the confidence intervals do not overlap starting from round $2$.
In terms of the empirical average, the improvement is
	obvious: at the end of the 5th round, for NetHEPT, MRIM algorithms is
	at least $8.8\%$ better than \alg{SG}, and $7.3\%$ better than \alg{SG-R};
	for Flixster, MRIM algorithms is
	at least $12.8\%$ better than \alg{SG}, and $4.3\%$ better than \alg{SG-R}.
Even if we use the ratio between the lower confidence bounds of MRIM algorithms vs.
	the upper confidence bounds of the baseline algorithms, the result is similar.
\alg{SG-R} performs better than \alg{SG}, which implies that influential nodes
	are important in these datasets and it is preferred to re-select them.
However, the improvement of MRIM algorithms over \alg{SG} and \alg{SG-R} are
	increasing over rounds, showing that adjusting to MRIM is increasingly
	important.
It is reasonable to see that with more rounds, always sticking to the same seed sets
	or always changing seed sets would not perform well.
	
Second, the adaptive algorithms perform better than the within-round non-adaptive algorithms
	in both dataset: the confidence interval does not overlap for all results
	in round $4$ and $5$ and most results in round $3$.
This conforms with our intuition that adaptive algorithms performs better.
The improvement are not very significant in 
	Tables ~\ref{table::hep_per} and~\ref{table::flixster_per}, which means the
	strength of the adaptiveness has not be fully explored yet.

Third, the cross-round setting is more effective, but less efficient than within-round setting of non-adaptive algorithms,
due to a larger search space.
In practice, the cross-round setting only shows outstanding performances overall, but doesn't guarantee the good performance in every round. In fact, it always performs worst in round $1$ comparing to the other algorithms. 

Fourth, the cross-round non-adaptive algorithm performs as well as the adaptive algorithm on NetHept dataset. 
Adaptive algorithm requires the real-life spread between each round comparing to non-adaptive algorithms, 
but always performs the best in each round.
	
Fifth, we can see that \alg{AdaIMM} achieves the same level of influence spread as \alg{AdaGreedy}: the confidence intervals always overlap, which means \alg{AdaIMM} performs well in practice.

\subsubsection{Degree of Adaptiveness}

We vary the parameters $T$ and $k$ whiling keeping $Tk$ the same.
With smaller $k$, it means each round we select a smaller number of seed sets, and 
	 we use more adaptive rounds.
Therefore, small $k$ and large $T$ mean a high degree of adaptiveness.
We test this using \alg{AdaIMM}, since it is more efficient than \alg{AdaGreedy}
	while providing the same level of influence spread.

Table \ref{table::tf} presents the result on this test, in which we vary 
	$(T,k)$ as $(1,50)$, $(2, 25)$, $(5,10)$ and $(10,5)$.
The influence spread significantly increases with the increase
	of the degree of adaptiveness, and the increase is quite significant.
The $10$-round $5$-seed setup is $36.3\%$ better than $1$ round $50$-seed setup on the empirical average.
This shows that higher adaptive degree indeed improves the performance.


\begin{table}[t]
	\centering
	\caption{Influence spread with different adaptive degree.}
	\label{table::tf}
	\resizebox{3.45in}{!}{%
		\begin{tabular}{|c|c|c|c|c|}
			\hline
			Num. of Rounds & 1                  & 2                    & 5                    & 10                                      \\ \hline
			Num. of Seeds  & 50                 & 25                   & 10                   & 5                                       \\ \hline
			\alg{AdaIMM}        & 883.0              & 1040.3               & 1141.0               & 1204.7                           \\ 
			(R = 100)        & {[}856.0, 910.1{]} & {[}1022.6, 1058.1{]} & {[}1119.3,  1162.6{]} & {[}1178.2, 1231.3{]}  \\ \hline
		\end{tabular}
	}
\end{table}

\begin{table}[]
	\centering
	\caption{Running time of the algorithms, in seconds.}
	\label{table:time}
	\resizebox{3.45in}{!}{%
		\begin{tabular}{lcccc}
			\hline
			\multicolumn{1}{|c|}{} & \multicolumn{1}{c|}{\alg{SG}} & \multicolumn{1}{c|}{\alg{SG-R}}              
			& \multicolumn{1}{c|}{\alg{E-WR-Greedy}}     & \multicolumn{1}{c|}{\alg{WR-IMM}}            \\ \hline
			\multicolumn{1}{|c|}{NetHEPT}  & \multicolumn{1}{c|}{439.2}        & \multicolumn{1}{c|}{87.8}      & \multicolumn{1}{c|}{551.2}  & \multicolumn{1}{c|}{1.97}              \\ 
			\multicolumn{1}{|c|}{(R = 5)}    & \multicolumn{1}{c|}{{[}407, 470.94{]}}     & \multicolumn{1}{c|}{{[}81.5, 94.2{]}}  & \multicolumn{1}{c|}{{[}527.9, 574.4{]}}   & \multicolumn{1}{c|}{{[}1.91, 2.03{]}} \\ \hline
			\multicolumn{1}{|c|}{Flixster} & \multicolumn{1}{c|}{4862.3}   & \multicolumn{1}{c|}{972.5}             & \multicolumn{1}{c|}{2478.9}               & \multicolumn{1}{c|}{3.16}             \\ 
			\multicolumn{1}{|c|}{(R = 5)}    & \multicolumn{1}{c|}{{[}4773.3, 4951.3{]}} & \multicolumn{1}{c|}{{[}990.3,954.7{]}} & \multicolumn{1}{c|}{{[}2422.4, 2535.5{]}} & \multicolumn{1}{c|}{{[}3.14, 3.18{]}} \\ \hline
			\\ \hline
			\multicolumn{1}{|c|}{}& \multicolumn{1}{c|}{\alg{CR-Greedy}} & \multicolumn{1}{c|}{\alg{CR-IMM}}    & \multicolumn{1}{c|}{\alg{AdaGreedy}}     & \multicolumn{1}{c|}{\alg{AdaIMM}}            \\ \hline
			\multicolumn{1}{|c|}{NetHEPT} & \multicolumn{1}{c|}{2105.6}              &      \multicolumn{1}{c|}{2.13}           & \multicolumn{1}{c|}{465.4}                & \multicolumn{1}{c|}{2.01}             \\ 
			\multicolumn{1}{|c|}{(R = 5)} & \multicolumn{1}{c|}{{[}2036.2, 2175.0{]}}  & \multicolumn{1}{c|}{{[}2.05, 2.21{]}}   & \multicolumn{1}{c|}{{[}473.8, 457.0{]}}   & \multicolumn{1}{c|}{{[}1.93, 2.09{]}} \\ \hline
			\multicolumn{1}{|c|}{Flixster} & \multicolumn{1}{c|}{9587.6}            &        \multicolumn{1}{c|}{3.61}        & \multicolumn{1}{c|}{2305.5}               & \multicolumn{1}{c|}{3.23}             \\ 
			\multicolumn{1}{|c|}{(R = 5)}  & \multicolumn{1}{c|}{{[}9145.3.10029.9{]}} & \multicolumn{1}{c|}{{[}3.59, 3.63{]}} & \multicolumn{1}{c|}{{[}2161.0, 2450.0{]}} & \multicolumn{1}{c|}{{[}3.16, 3.30{]}} \\ \hline
		\end{tabular}
	}
\end{table}

In summary, with same total budget, the higher number of total rounds with higher influence spread performance in general.
\vspace{-0.1in}
\subsubsection{Running Time}

Table \ref{table:time} reports the running time of all the tested algorithms on the
	two datasets, when running with $T=5$ and $k=10$.
One clear conclusion is that all \alg{IMM} algorithms are much more efficient that others,
	with two to three orders of magnitude faster than all other algorithms.
Among greedy algorithms, 
\alg{SG-R} runs faster because it selects only $10$ seeds once.
\alg{SG} is the slowest, much slower than \alg{E-WR-Greedy} and \alg{AdaGreedy}
	because it needs to select $50$ different seeds, and when it selects more seeds,
	their marginal influence spread does  not differ from one another much and thus
	the lazy-evaluation optimization is not as effective as selecting the first $10$
	seeds.


After combining the influence spread and running time performance, our conclusion
	is that (a) algorithms designed for the MRIM task is better,
	(b) the cross-round setting is more effective, but less efficient than the within-round setting of non-adaptive algorithms,
	and (c) \alg{AdaIMM} is clearly the best for adaptive MRIM task.
	
\section{Acknowledge}

This work is supported in part by NSF through grants IIS-1526499, IIS-1763325, and CNS-1626432, and NSFC 61672313, 61433014.

\vspace{-4.5pt}

%


\bibliographystyle{unsrtnat}
\bibliography{bibdatabase} 

\begin{thebibliography}{36}
\providecommand{\natexlab}[1]{#1}
\providecommand{\url}[1]{\texttt{#1}}
\expandafter\ifx\csname urlstyle\endcsname\relax
  \providecommand{\doi}[1]{doi: #1}\else
  \providecommand{\doi}{doi: \begingroup \urlstyle{rm}\Url}\fi

\bibitem[Berger(2014)]{berger2014word}
Jonah Berger.
\newblock Word of mouth and interpersonal communication: A review and
  directions for future research.
\newblock \emph{Journal of Consumer Psychology}, 2014.

\bibitem[Domingos and Richardson(2001)]{domingos01}
Pedro Domingos and Matthew Richardson.
\newblock Mining the network value of customers.
\newblock In \emph{KDD}, 2001.

\bibitem[Richardson and Domingos(2002)]{richardson02}
Matthew Richardson and Pedro Domingos.
\newblock Mining knowledge-sharing sites for viral marketing.
\newblock In \emph{KDD}, 2002.

\bibitem[Kempe et~al.(2003)Kempe, Kleinberg, and Tardos]{kempe03}
David Kempe, Jon~M. Kleinberg, and {\'E}va Tardos.
\newblock Maximizing the spread of influence through a social network.
\newblock In \emph{KDD}, 2003.

\bibitem[Leskovec et~al.(2007)Leskovec, Krause, Guestrin, Faloutsos,
  VanBriesen, and Glance]{Leskovec07}
Jure Leskovec, Andreas Krause, Carlos Guestrin, Christos Faloutsos, Jeanne~M.
  VanBriesen, and Natalie~S. Glance.
\newblock Cost-effective outbreak detection in networks.
\newblock In \emph{KDD}, 2007.

\bibitem[Budak et~al.(2011)Budak, Agrawal, and Abbadi]{BAA11}
Ceren Budak, Divyakant Agrawal, and Amr~El Abbadi.
\newblock {Limiting the spread of misinformation in social networks}.
\newblock In \emph{{WWW11}}, 2011.

\bibitem[He et~al.(2012)He, Song, Chen, and Jiang]{HeSCJ12}
Xinran He, Guojie Song, Wei Chen, and Qingye Jiang.
\newblock {Influence Blocking Maximization in Social Networks under the
  Competitive Linear Threshold Model}.
\newblock 2012.

\bibitem[Chen et~al.(2013)Chen, Lakshmanan, and Castillo]{chen2013information}
Wei Chen, Laks~VS Lakshmanan, and Carlos Castillo.
\newblock \emph{Information and Influence Propagation in Social Networks}.
\newblock Morgan \& Claypool Publishers, 2013.

\bibitem[Golovin and Krause(2011)]{GoloKrause11}
Daniel Golovin and Andreas Krause.
\newblock Adaptive submodularity:theory and applications in active learning and
  stochastic optimization.
\newblock \emph{JAIR}, 2011.

\bibitem[Borgs et~al.(2014)Borgs, Brautbar, Chayes, and
  Lucier]{BorgsBrautbarChayesLucier}
Christian Borgs, Michael Brautbar, Jennifer Chayes, and Brendan Lucier.
\newblock Maximizing social influence in nearly optimal time.
\newblock In \emph{ACM-SIAM}, SODA '14, 2014.

\bibitem[Tang et~al.(2014)Tang, Xiao, and Shi]{tang14}
Youze Tang, Xiaokui Xiao, and Yanchen Shi.
\newblock Influence maximization: near-optimal time complexity meets practical
  efficiency.
\newblock In \emph{SIGMOD}, 2014.

\bibitem[Tang et~al.(2015)Tang, Shi, and Xiao]{tang15}
Youze Tang, Yanchen Shi, and Xiaokui Xiao.
\newblock Influence maximization in near-linear time: a martingale approach.
\newblock In \emph{SIGMOD}, 2015.

\bibitem[Chen et~al.(2009)Chen, Wang, and Yang]{ChenWY09}
Wei Chen, Yajun Wang, and Siyu Yang.
\newblock Efficient influence maximization in social networks.
\newblock In \emph{KDD}, 2009.

\bibitem[Chen et~al.(2010)Chen, Yuan, and Zhang]{ChenYZ10}
Wei Chen, Yifei Yuan, and Li~Zhang.
\newblock Scalable influence maximization in social networks under the linear
  threshold model.
\newblock In \emph{ICDM}, 2010.

\bibitem[Wang et~al.(2012)Wang, Chen, and Wang]{WCW12}
Chi Wang, Wei Chen, and Yajun Wang.
\newblock {Scalable influence maximization for independent cascade model in
  large-scale social networks}.
\newblock \emph{DMKD}, 2012.

\bibitem[Jung et~al.(2012)Jung, Heo, and Chen]{JungHC12}
Kyomin Jung, Wooram Heo, and Wei Chen.
\newblock {IRIE: Scalable and Robust Influence Maximization in Social
  Networks}.
\newblock In \emph{ICDM}, 2012.

\bibitem[Tantipathananandh et~al.(2007)Tantipathananandh, Berger-Wolf, and
  Kempe]{kempe07}
Chayant Tantipathananandh, Tanya Berger-Wolf, and David Kempe.
\newblock A framework for community identification in dynamic social networks.
\newblock In \emph{KDD'07}, 2007.

\bibitem[Chen et~al.(2011)Chen, Collins, Cummings, Ke, Liu, Rincon, Sun, Wang,
  Wei, and Yuan]{chen2011influence}
Wei Chen, Alex Collins, Rachel Cummings, Te~Ke, Zhenming Liu, David Rincon,
  Xiaorui Sun, Yajun Wang, Wei Wei, and Yifei Yuan.
\newblock Influence maximization in social networks when negative opinions may
  emerge and propagate.
\newblock In \emph{SDM}, 2011.

\bibitem[Lu et~al.(2015)Lu, Chen, and Lakshmanan]{lu2015competition}
Wei Lu, Wei Chen, and Laks~VS Lakshmanan.
\newblock From competition to complementarity: comparative influence diffusion
  and maximization.
\newblock \emph{PVLDB}, 2015.

\bibitem[He et~al.(2016)He, Lu, Ma, Cao, Shen, and Yu]{he2016joint}
Lifang He, Chun-Ta Lu, Jiaqi Ma, Jianping Cao, Linlin Shen, and Philip~S Yu.
\newblock Joint community and structural hole spanner detection via harmonic
  modularity.
\newblock In \emph{SIGKDD}. ACM, 2016.

\bibitem[Bhagat et~al.(2012)Bhagat, Goyal, and
  Lakshmanan]{bhagat_2012_maximizing}
Smriti Bhagat, Amit Goyal, and Laks V.~S. Lakshmanan.
\newblock {Maximizing product adoption in social networks}.
\newblock In \emph{WSDM}. ACM, 2012.

\bibitem[Chen et~al.(2016{\natexlab{a}})Chen, Lin, Tan, Zhao, and
  Zhou]{ChenLTZZ16}
Wei Chen, Tian Lin, Zihan Tan, Mingfei Zhao, and Xuren Zhou.
\newblock Robust influence maximization.
\newblock In \emph{KDD}, 2016{\natexlab{a}}.

\bibitem[{He} and {Kempe}(2016)]{HeKempe16}
X.~{He} and D.~{Kempe}.
\newblock Robust influence maximization.
\newblock In \emph{KDD}, 2016.

\bibitem[Seeman and Singer(2013)]{seeman2013adaptive}
Lior Seeman and Yaron Singer.
\newblock Adaptive seeding in social networks.
\newblock In \emph{FOCS}. IEEE, 2013.

\bibitem[Lin et~al.(2015)Lin, Lin, and Chen]{lin2015learning}
Su-Chen Lin, Shou-De Lin, and Ming-Syan Chen.
\newblock A learning-based framework to handle multi-round multi-party
  influence maximization on social networks.
\newblock In \emph{KDD}. ACM, 2015.

\bibitem[Lei et~al.(2015)Lei, Maniu, Mo, Cheng, and Senellart]{LeiMMCS15}
Siyu Lei, Silviu Maniu, Luyi Mo, Reynold Cheng, and Pierre Senellart.
\newblock Online influence maximization.
\newblock In \emph{KDD}, 2015.

\bibitem[Chen et~al.(2016{\natexlab{b}})Chen, Wang, Yuan, and Wang]{CWYW16}
Wei Chen, Yajun Wang, Yang Yuan, and Qinshi Wang.
\newblock Combinatorial multi-armed bandit and its extension to
  probabilistically triggered arms.
\newblock \emph{Journal of Machine Learning Research}, 2016{\natexlab{b}}.

\bibitem[Vaswani et~al.(2015)Vaswani, Lakshmanan, and Schmidt]{VaswaniL15}
Sharan Vaswani, Laks V.~S. Lakshmanan, and Mark Schmidt.
\newblock Influence maximization with bandits.
\newblock In \emph{NIPS Workshop}, 2015.

\bibitem[Wen et~al.(2017)Wen, Kveton, Valko, and Vaswani]{Wen2016}
Zheng Wen, Branislav Kveton, Michal Valko, and Sharan Vaswani.
\newblock Online influence maximization under independent cascade model with
  semi-bandit feedback.
\newblock In \emph{NIPS}, 2017.

\bibitem[Vaswani et~al.(2017)Vaswani, Kveton, Wen, Ghavamzadeh, Lakshmanan, and
  Schmidt]{VKWGLS17}
Sharan Vaswani, Branislav Kveton, Zheng Wen, Mohammad Ghavamzadeh, Laks~V.S.
  Lakshmanan, and Mark Schmidt.
\newblock Diffusion independent semi-bandit influence maximization.
\newblock In \emph{ICML}, 2017.

\bibitem[Nemhauser et~al.(1978)Nemhauser, Wolsey, and Fisher]{NWF78}
G.~L. Nemhauser, L.~A. Wolsey, and M.~L. Fisher.
\newblock \emph{An analysis of the approximations for maximizing submodular set
  functions}.
\newblock Mathematical Programming, 1978.

\bibitem[Fisher et~al.(1978)Fisher, Nemhauser, and Wolsey]{fisher1978analysis}
Marshall~L Fisher, George~L Nemhauser, and Laurence~A Wolsey.
\newblock An analysis of approximations for maximizing submodular set functions
  ii.
\newblock In \emph{Polyhedral combinatorics}. Springer, 1978.

\bibitem[Goundan and Schulz(2007)]{GS07}
P.~R. Goundan and A.~S. Schulz.
\newblock Revisiting the greedy approach to submodular set function
  maximization.
\newblock Technical report, MIT, 2007.

\bibitem[Mossel and Roch(2007)]{mossel2007}
Elchanan Mossel and Sebastien Roch.
\newblock On the submodularity of influence in social networks.
\newblock In \emph{STOC '07}, 2007.

\bibitem[Barbieri et~al.(2012)Barbieri, Bonchi, and Manco]{barbieri2012topic}
Nicola Barbieri, Francesco Bonchi, and Giuseppe Manco.
\newblock Topic-aware social influence propagation models.
\newblock In \emph{ICDM'12}. IEEE, 2012.

\bibitem[Chen(2008)]{chen08}
Ning Chen.
\newblock On the approximability of influence in social networks.
\newblock In \emph{SODA '08}. Society for Industrial and Applied Mathematics,
  2008.

\end{thebibliography}
\clearpage

\OnlyInFull{
\appendix

\section*{Appendix}
The appendix is organized as follows.
In Section~\ref{app:proofs}, we give all proofs of lemmas and theorems , i.e.,
    Section~\ref{app:namrim} for the non-adaptive MRIM, 
    Section~\ref{app:amrim} for the adaptive MRIM,
    Section~\ref{app:crimm} for the \CRIMM,
    and Section~\ref{app:wrimm} for the \WRIMM,
In Section~\ref{app:pseudocode}, we give our complete IMM algorithms and compare them with the original IMM.
\section{IMM Algorithms} \label{app:pseudocode}

\subsection{Cross-Round IMM Algorithm}

\alg{\CRIMM} is very similar to standard \alg{IMM} and only has a few differences include several points.
We have talked about them in the formal paper already.

\subsection{Within-Round IMM Algorithm} \label{app:WRIMM}


\alg{\WRIMM} (Algorithm \ref{alg:wrimm})  contains two main phases like standard \alg{IMM}.
In Phase 1, the algorithm estimates the number of RR sets $\theta$  needed by estimating
the lower bound $LB$ of the optimal influence spread.
It is done by iteratively halving a threshold value $x_t$, and for each
$x_t$ value, generating $\theta_j$ RR sets (lines~\ref{line:wgenRRset1b}--\ref{line:wgenRRset1e}),
and finding the seed set $\cS_j$ from the current RR sets (line~\ref{line:wnodeselect1}), 
and checking if the current $x_t$ is good enough for the lower bound estimate
(line~\ref{line:wcheckLB}), where $F_{M}(\cS_j)$ is the fraction of RR sets
in $\cM$ that are covered by $\cS_j$.
Procedure $\alg{NodeSelection}(\cM,k)$ is to greedily select $k$ nodes that covers
the most number of RR sets in $\cM$, and it is the same as in~\cite{tang15}, so it
is omitted.
In Phase 2, based on the $\theta$ RR sets generated, $\cS_t$ can be selected
by procedure $\alg{NodeSelection}(\cM,k)$.
Finally, propagation from $\cS_t$ are observed and the active nodes are used as the
feedback for the next round.

\begin{algorithm}[h] 
	\caption{{\WRNS}: Within-Round Node Selection} \label{alg:wrnr}
	\KwIn{Multi round RR vector sets $\cM$, $t$, $k$}
	\KwOut{seed sets $\cS_t^o$}
	
	Build count array: $c[(u, t)] = \sum_{(u, t) \in \cM}|(u, t)|$, $\forall (u, t) \in \cV_t$\; \label{line:wrgen1}
	Build RR set link: $RR[(u, t)]$, $\forall (u, t) \in \cV_t$\; \label{line:wrgen2}
	For all $\cR \in \cM$, $covered[\cR] = false$\;
	$\cS_t^o \leftarrow \emptyset$; $\cC\leftarrow \cV_t$\;
	\For{$i = 1 \; to \; k$ \label{line:wrfor1}}{
		$(u, t) \leftarrow \argmax_{(u', t)\in \cC\setminus \cS_t^o} c[(u', t)]$\;\label{line:wrf2}
		$\cS_t^o \leftarrow \cS_t^o \cup \{(u, t)\}$\; 
		\For{all $\cR \in RR[(u, t)] \wedge covered[\cR]= false$}{\label{line:wrfor2}
			$covered[\cR]= true$\;
			for all $(u', t) \in \cR \wedge (u', t) \ne (u, t)$ to do $c[(u', t')] = c[(u', t')] -1$\;\label{line:wrfor3}
		}
	}
	\bf{return} $\cS_t^o $.
\end{algorithm}

%

\begin{algorithm}[t] 
	\caption{{\WRIMM}: Non-adaptive IMM Algorithm for Within Round} \label{alg:wrimm}
	\KwIn{Graph $G=(V,E)$, round number $T$, budget $k$, accuracy parameters $(\varepsilon, \ell)$, triggering set distributions}
	\KwOut{$\cS^o = \cS_1^o \cup \cS_2^o \cup \dots\cup \cS_T^o$.}
	
	$\cS_1^o, \cS_2^o, \dots, \cS_T^o \leftarrow \emptyset$;
	$\ell \leftarrow \ell + \ln (2T) / \ln n $;
	$LB \leftarrow 1$; $\varepsilon_0 = e^{(1-1/e)} \varepsilon/2$; $\varepsilon' = \sqrt{2}\varepsilon_0$\;
	$\alpha \leftarrow \sqrt{\ell \ln{n} + \ln{2} + \ln T}$;
	$\beta \leftarrow \sqrt{(1-1/e)\cdot (\ln{\binom{n}{k}}+\alpha^2)}$\;
	$\lambda' \leftarrow [(2+\frac{2}{3}\varepsilon') \cdot (\ln{\binom{n}{k}}+\ell \cdot \ln{n}+ \ln T+\ln{\log_2{n}})\cdot n]/\varepsilon'^2$\;
	$\lambda^* \leftarrow 2 n \cdot ((1-1/e) \cdot \alpha + \beta)^2 \cdot \varepsilon_0^{-2}$\;
	$Root_0 = V$\;
	\For{$t=1$ to $T$}{
		\tcp{Phase 1: Estimating the number of RR sets needed, $\theta$}
		$\cM \leftarrow \emptyset$\;
		\For{$j = 1$  to  $\log_2{(\newalg{n} - 1)}$ \label{line:wfor1}}{
			$x_t \leftarrow \newalg{n}/2^j$\; \label{line:wassignx}
			$\theta_j \leftarrow \lambda'/x_t$\; \label{line:wgenRRset1b}
			\While{$|\cM| < \theta_j$}{
				Select a node $u$ from $Root_{t-1}$ uniformly at random \; \label{line:wsample1}
				Generate RR set with pair notation $\cR_{u,t}$ from $u$, and insert it into $\cM$\; \label{line:wgsample1}
			}\label{line:wgenRRset1e} 
			
			$\cS_j \leftarrow \NodeSelection(\cM,k)$\; \label{line:wnodeselect1}
			\If{$\newalg{ n \cdot F_{\cM}(\cS_j)} \geq (1 + \varepsilon') \cdot x$}{
				\label{line:westimate1}
				$LB \leftarrow \newalg{ n \cdot F_{\cM}(\cS_j)}/(1 + \varepsilon')$\; 
				\label{line:westimate2}
				break\;
			} 
		}
		
		$\theta \leftarrow \lambda^* / LB$\;
		\While{$|\cM| \leq \theta$}{ \label{line:wcheckLB}
			Select a node $u$ from $Root_{t-1}$ uniformly at random\;\label{line:wsample2}
			Generate $\cR_{u,t}$ in round $t$, and insert it into $\cM$\; 	\label{line:wgsample2}	
		}
		
		\tcp{Phase 2: Generate $\theta$ RR-vector sets and select seed nodes}
		$\cS^o_t \leftarrow \NodeSelection(\cM, k)$\;
		$Root_t \gets$ the remaining roots of the RR sets in $\cM$ after the $\NodeSelection$ procedure, treat it as a multiset
	}
	\bf{return} $\cS^o = \cS_1^o \cup \cS_2^o \cup \dots\cup \cS_T^o$
\end{algorithm}

\subsection{Adaptive IMM Algorithm}

\begin{algorithm}[t] 
	\caption{{\alg{AdaIMM}}: Adaptive IMM algorithm for round $t$ } \label{alg:aimm}
	\KwIn{Graph $G=(V,E)$, round number $T$, budget $k$, accuracy parameters $(\varepsilon, \ell)$, triggering set distributions, all active nodes by round $t-1$ $A_{t-1}$}
	\KwOut{seed set $S_i$, \newalg{and updated active nodes $A_i$}}

	\tcp{Phase 1: Estimating the number of RR sets needed, $\theta$}
	
	
	\newalg{$n_a \leftarrow n - |A_{t-1}|$}; 
	\newalg{$\ell \leftarrow \ell + \ln (2T) / \ln n $};
	\newalg{$\varepsilon_0 = e^{(1-1/e)} \varepsilon/2$}\;
	$\varepsilon' \leftarrow \sqrt{2} \cdot \varepsilon_0 $; 
	$\cR_{t} \leftarrow \emptyset$; $M \leftarrow \emptyset$\; 
	$\alpha \leftarrow \sqrt{\ell \ln{n} + \ln{2}}$;
	$\beta \leftarrow \sqrt{(1-1/e)\cdot (\ln{\binom{n}{k}}+\alpha^2)}$\;
	$\lambda' \leftarrow [(2+\frac{2}{3}\varepsilon') \cdot (\ln{\binom{n}{k}}+\ell \cdot \ln{n}+\ln{\log_2{n}})\cdot n]/\varepsilon'^2$\;
	$\lambda^* \leftarrow 2 n \cdot ((1-1/e) \cdot \alpha + \beta)^2 \cdot \varepsilon_0^{-2}$\;
	\For{$j = 1 \; to \; \log_2{(\newalg{n_a} - 1)}$ \label{line:for1}}{
		$x_t \leftarrow \newalg{n_a}/2^j$\; \label{line:assignx}
		$\theta_j \leftarrow \lambda'/x_t$\; \label{line:genRRset1b} 
		\While{$|M| < \theta_j$}{
			Select a node $v$ from \newalg{$V \setminus A_{t-1}$} uniformly at random\; \label{line:sample1}
			Generate RR set $R$ from $v$, and insert it into $M$\; 
		}\label{line:genRRset1e} 
		$S_j \leftarrow \NodeSelection(M,k)$\; \label{line:nodeselect1}
		\If{$\newalg{n_a} \cdot F_{M}(S_j) \geq (1 + \varepsilon') \cdot x_i$}{
			\label{line:estimate1}
			$LB \leftarrow \newalg{n_a} \cdot F_{M}(S_j)/(1 + \varepsilon')$\; 
			\label{line:estimate2}
			break\;
		} 
	}
	$\theta \leftarrow \lambda^* / LB$\;
	\While{$|M| \leq \theta$}{ \label{line:checkLB}
		Select a node $v$ from \newalg{$V \setminus A_{t-1}$} uniformly at random\;\label{line:sample2}
		Generate $R$ for $v$, and insert it into $M$\; 		
	}
	
	\tcp{Phase 2: Generate $\theta$ RR sets and select seed nodes}
	$S_t \leftarrow \NodeSelection(M,k)$ \;
	\newalg{Observe the propagation from $S_t$ in round $t$, update
		the set of activated nodes to $A_t$}\; 
	
	\bf{return} $S_t$, \newalg{$A_t$}.
\end{algorithm}

Therefore, with Eq.~\eqref{eq:AdaRRset}, the same RR-set based algorithm can be used,
and we only need to properly change the RR-set generation process and the
estimation process.
Algorithm~\ref{alg:aimm} provides the pseudocode for the \alg{AdaIMM} algorithm, which
is based on the \alg{IMM} algorithm in~\cite{tang15}.
For convenience, we mark the part different from \alg{IMM} in blue.
Phase 1 of the algorithm estimate the number $\theta$ of RR sets needed by estimating
the lower bound $LB$ of the optimal spread value.
This estimation is by iteratively halving a threshold value $x_t$, and for each
$x_t$ value, generate $\theta_j$ RR sets (lines~\ref{line:genRRset1b}--\ref{line:genRRset1e}),
and find the seed set $S_j$ from the current RR sets (line~\ref{line:nodeselect1}), 
and check
if the current $x_t$ is good enough for the lower bound estimate
(line~\ref{line:checkLB}), where $F_{M}(S_j)$ is the fraction of RR sets
in $M$ that are covered by $S_j$.
Procedure $\alg{NodeSelection}(M,k)$ is to greedily select $k$ nodes that covers
the most number of RR sets in $M$, and it is the same as in~\cite{tang15}, so it
is omitted.
In Phase 2, $\theta$ RR sets are generated and node $S_t$ is selected greedily
by procedure $\alg{NodeSelection}(M,k)$.
Finally, propagation from $S_t$ are observed and the activated nodes are used as the
feedback for the next round.

The main differences from the standard \alg{IMM} include several points.
First, whenever we generate new RR sets in round $t$, we only start from roots in
$V \setminus A_{t-1}$ (lines~\ref{line:sample1} and~\ref{line:sample2}), as explained by Lemma~\ref{lem:ARRset}.
Second, when we estimate the influence spread, we need to adjust it using
$n_a = n - |A_{t-1}|$ (lines~\ref{line:for1}, \ref{line:assignx}, \ref{line:estimate1}, \ref{line:estimate2}),
again by Lemma~\ref{lem:ARRset}.
Third, we need to adjust $\ell$ to be $\ell + \log(2T)/\log n$,
and $\varepsilon$ to $\varepsilon_0 = e^{(1-1/e)} \varepsilon/2$.
This is to guarantee that in each round we have probability at least $1 - 1/(n^\ell T)$ to
have $S_t$ as a $(1-1/e - \varepsilon_0)$ approximation, so that the result for the
total $T$ rounds would come out correctly as stated in the following theorem.

\section{Proofs for Lemmas and Theorems}\label{app:proofs}

\subsection{Proofs for Section \ref{sec:nonadaptive} [Non-Adaptive MRIM]}\label{app:namrim}

\crmrim*
\begin{proof}[Proof (Sketch)]
	
	Given a set $U$ which is partitioned into disjoint sets $U_1, \dots, U_n$ and $\mathcal{I}=\{X\subseteq U\colon |X\cap U_i|\le k_i, \forall i\in[n]\}$, $(U,\mathcal{I})$ is called a {\em partition matroid}.
	Thus, the node space $\cV$ with the constraint of MRIM, namely $(\cV, \{\cS\colon|S_t|\le k, \forall t\in[T]\})$, is a partition matroid.
	This indicates that MRIM under cross-round setting is an instance of submodular maximization under partition matroid, 
	and thus the performance of \alg{CR-Greedy} has the following guarantee \cite{fisher1978analysis}. Meanwhile, Lemma~\ref{lem:mrtsubmodular} shows that $\rho$ is monotone and submodular in the cross-rounding setting. 
	Therefore, the final output $\cS^o$	 satisfies:
	
	\begin{equation*}
	\rho(\cS^o) \ge (\frac{1}{2}-\varepsilon) \rho(\cS^*).
	\end{equation*}
	While $R=\lceil 31 k^2T^2n\log(2kn^{\ell+1} T)/\varepsilon^2 \rceil$.
	Finally, the total running time is simply 
	$O(TknRm) = O(k^3\ell T^3n^2m\log(n T)/\varepsilon^2)$.
\end{proof}

\subsection{Proof for Section \ref{sec:AMRIM} [Adaptive MRIM]}\label{app:amrim}

This is the proof of adaptive sub-modularity Lemma~\ref{lem:adaSubmodular}.

\adamonotone*

\adasubmodular*

\begin{proof}     
	For all $\psi$ and $\psi'$, such that $\psi \subseteq \psi'$ and $|\dom(\psi)| = i'-1$.
	Let $i=|\dom(\psi)| +1$, then $i\le i'$.
	We want to show that $\Delta((S_{i'},i') | \psi) \geq  \Delta((S_{i'},i') | \psi')$.
	To prove this rigorously, we define a coupled distribution $\mu$ over pairs of realizations $\phi \sim \psi$
	and $\phi' \sim \psi'$. 
	Let $L^\phi_t$ represent the live-edge graph in round $t$ that determines
	the propagation.
	Note that in the MRT model, each realization $\phi$ corresponds to a live-edge graph sequence $\{(L_1^\phi, \dots, L_T^\phi)\}$.
	We use $\bX$ to denote the random live-edge graph sequence $\{(L_1, \dots, L_T)\}$.
	We define $\mu$ implicitly in terms of a joint distribution on $\bX \times \bX'$, 
	where $\phi = \phi(\bX)$ and $\phi' = \phi'(\bX')$ are the realizations induced by the two distinct sets
	of random live-edge graphs, respectively. Hence, $\mu(\phi(\bX),\phi(\bX')) = \hat{\mu}(\bX,\bX')$.
	We will construct $\hat{\mu}$ so that the status of live-edge graph which are unobserved by
	both $\psi$ and $\psi'$ and are the same in $\bX$ and $\bX'$, 
	meaning $L_i = L_i'$ for all such live-edge graphs, or else $\hat{\mu}(\bX,\bX') = 0$.
	
	The above constraints leave us with the following degrees of freedom: we may select $X_i$
	which is unobserved by $\psi$.
	We select them independently. Hence for all $(X, X')$ satisfying the above constraints, where
	$X(e, L_j(\phi))$ shows the status of $e \in E$ in the $j$-th round with $L_j$ under $\phi$
	
	
	
	\begin{align}
	\hat{\mu}(\bX,\bX') 
	&= \prod_{j \geq i} \Pr(L_j^\phi), \label{eq:17}
	\end{align}   
	and otherwise $\mu(X,X') = 0$. Note that only $\phi$ will influence the
	$\hat{\mu}(\bX,\bX')$. 
	
	We first claim that
	given a partial realization $\psi$ with $i$ rounds, and
	a particular realization $\phi$,
	the conditional probability of $\phi$ given $\psi$, 
	$p(\phi | \psi )$, satisfies that
	\begin{equation*}
	p(\phi | \psi ) = \sum_{\phi'} \mu(\phi ,\phi'),
	\end{equation*}
	where $\phi'$ is another particular realization, $\psi \subseteq \psi'$
	and $\mu(\cdot)$ is the joint distribution. The proof is as follows.
	
	First, we have fixed $\psi$ and $\psi'$, where $\psi \subseteq \psi'$.
	By definition of a coupled distribution $\mu$, if and only if 
	$\phi \sim \psi$ and $\phi' \sim \psi'$, $\mu(\cdot) \neq 0$. 
	We want to calculate the conditional probability of $\phi$ given $\psi$.
	For any fixed $\phi$, let a set $(S_t, t) \in \dom(\psi')$. 
	When $t \leq i' - 1$, if we want $\mu(\cdot) \neq 0$, 
	we have $\phi' \sim \psi'$ and $\phi'((S_t, t)) = \psi'((S_t, t))$.
	When $t \geq i'$, by definition, if we want $\mu(\cdot) \neq 0$,
	$\phi'((S_t, t)) = \phi((S_t, t))$. 
	Since $\phi$ is fixed, $\phi'$ is consistent with $\phi$ when $t \geq i'$.
	Finally, we can only find one $\phi'$ that satisfies the constraints for
	non-zero $\mu(\phi,\phi')$, 
	and we have $ \sum_{\phi'} \mu(\phi ,\phi') = \mu(\phi,\phi')
	=\hat{\mu}(\phi(X),\phi'(X')) = \prod_{j \geq i} \Pr(L_j^\phi) = p(\phi \mid \psi)$,
	where the last equality holds due to the definition of $\phi$ and
	$\psi$.

	We then claim that
	given a partial realization $\psi'$ with $i$ rounds, and
	a particular realization $\phi'$,
	the conditional probability of $\phi'$ given $\psi'$, 
	$p(\phi' | \psi' )$, satisfies that
	\begin{equation*}
	p(\phi' | \psi' ) = \sum_{\phi} \mu(\phi ,\phi'),
	\end{equation*}
	where $\phi$ is the other a particular realization, $\psi \in \psi'$
	and $\mu(\cdot)$ is the joint distribution.
	We prove this as follows.

	Again, we have fixed $\psi$ and $\psi'$, where $\psi \subseteq \psi'$.
	We want to calculate the conditional probability of $\phi'$ given $\psi'$, where $\phi'$ is a fixed realization.
	As the same, if and only if $\phi \sim \psi$ and $\phi' \sim \psi'$,
	$\mu(\cdot) \neq 0$. 
	Consider a set $(S_t, t) \in \dom(\psi')$. 
	When $t \leq i - 1$,  if we want $\mu(\cdot) \neq 0$, 
	$\phi((S_t, t)) = \psi((S_t, t))$.
	When $t \geq i'$,  if we want $\mu(\cdot) \neq 0$, 
	$\phi((S_t, t)) = \phi'((S_t, t))$.
	When $ i \le t \le i'-1$, we want $\mu(\cdot) \neq 0$ with any $\phi$.
	Finally, we can only find multiple $\phi$ satisfies the constraints for
	non-zero $\mu(\phi,\phi')$,
	and we have $\sum_{\phi} \mu(\phi ,\phi')
	=\sum_{\phi} \hat{\mu}(\phi(X),\phi'(X')) 
	= \sum_{\phi} \prod_{j \geq i} \Pr(L_j^\phi)
	= \prod_{j \ge i'} \Pr(L_j^{\phi'}) \sum_{\phi} \prod_{i \le j \le i'-1} \Pr(L_j^\phi)
	= \prod_{j \ge i'} \Pr(L_j^{\phi'})  =
	p(\phi' | \psi')
	$.

	We next claim that for all $(\phi, \phi') \in \text{support}(\mu)$,
	\begin{align*}
	& f(\dom(\psi') \cup (S_{i'},i') , \phi') - f(\dom(\psi'), \phi') \\
	& \le f(\dom(\psi) \cup (S_{i'},i') , \phi) - f(\dom(\psi), \phi). 
	\end{align*}
	A sufficient condition for the above inequality is
	\begin{align*}
	\Gamma(L_{i'}^{\phi'}, S_{i'}) \setminus \bigcup_{j=1}^{i'-1} \Gamma(L_j^{\phi'}, S_j)  \subseteq 
	\Gamma(L_{i'}^\phi, S_{i'}) \setminus \bigcup_{j=1}^{i-1} \Gamma(L_j^\phi, S_j). 
	\end{align*}
	Let $C' = \Gamma(L_{i'}^{\phi'}, S_{i'})$, $B' = \Gamma(L_{i'}^\phi, S_{i'})$
	denote the active before and after select $S_{i'}$ as the beginning nodes
	at $i'$ round respect to $\psi'$ and $\phi'$. 
	Similarity, define $C$ a $B$ respect to $\psi$ and $\phi$. 
	Let $D = C \setminus B$ and $D' = C' \setminus B'$.
	It suffices to show that $B \subseteq B'$ and $D' \subseteq D$ to prove the above inequality, which we will now do.
	By definition, we know that $C' = C$. 
	Based on adaptive monotonicity, we have $B \subseteq B'$. 
	Then, it is very obviously to prove that $D' \subseteq D$.
	
	
	
	
	Now we now proceed to use it to show $\Delta((S_{i'},i') | \psi) \geq  \Delta((S_{i'},i') | \psi')$.
	\begin{align*}
	&\Delta((S_{i'},i') | \psi') \\
	&= \sum_{\phi'}p(\phi'|\psi') f(\dom(\psi') \cup (S_{i'},i') , \phi') - f(\dom(\psi'), \phi'))\\
	&= \sum_{\phi'}\sum_{\phi} \mu(\phi ,\phi') f(\dom(\psi') \cup (S_{i'},i') , \phi') - f(\dom(\psi'), \phi'))\\
	&= \sum_{\phi, \phi'} \mu(\phi, \phi') f(\dom(\psi') \cup (S_{i'},i') , \phi') - f(\dom(\psi'), \phi'))\\
	&\leq \sum_{\phi, \phi'} \mu(\phi, \phi') f(\dom(\psi) \cup (S_{i'},i') , \phi) - f(\dom(\psi), \phi))\\
	&= \sum_{\phi}\sum_{\phi'} \mu(\phi ,\phi') f(\dom(\psi) \cup (S_{i'},i') , \phi) - f(\dom(\psi), \phi))\\
	&= \sum_{\phi}f(\dom(\psi) \cup (S_{i'},i') , \phi) - f(\dom(\psi), \phi))\sum_{\phi'} \mu(\phi ,\phi') \\
	&= \sum_{\phi}p(\phi|\psi) f(\dom(\psi) \cup (S_{i'},i') , \phi) - f(\dom(\psi), \phi))\\
	&= \Delta((S_{i'},i') | \psi)
	\end{align*}
	This finishes the proof of the adaptive submodularity.
\end{proof}

\adagreedy*
\begin{proof}[Proof (Sketch)]
	From the general result of Theorem 5.2 in~\cite{GoloKrause11}, we know that
	if in each round we find an $\alpha$ approximation of the best solution of the round,
	then the greedy adaptive policy is an $1-e^{-\alpha}$ approximation of the optimal
	greedy policy $\pi^*$.
	Then, following essentially the same high probability argument as in the proof sketch of 
	Theorem~\ref{thm:nonadaptive}, we could conclude that
	when using $R=\lceil 31 k^2n\log(2kn^{\ell+1} T)/\varepsilon^2 \rceil$
	Monte Carlo simulations for each influence spread estimation, with probability
	at least $1-\frac{1}{n^\ell}$, in all rounds the selected $S_i$ 
	is an $1-1/e - \varepsilon_1$ approximation, with $\varepsilon_1 = e^{(1-1/e)} \varepsilon/2$.
	Together with Theorem 5.2 in~\cite{GoloKrause11}, we have that
	with probability at least $1-\frac{1}{n^\ell}$, 
	$f_{\avg}(\pi^{\ag}) \geq \left(1 - e^{-(1-\frac{1}{e})}-\varepsilon \right) f_{\avg}(\pi^*)$.
	The running time is $O(TknRm) = O(k^3\ell Tn^2m\log(n T)/\varepsilon^2)$.
\end{proof}


\subsection{Proofs for Section \ref{sec:CRIMM} [\CRIMM]}\label{app:crimm}

\crimm*
\begin{proof}
	
	\begin{align}
	& \E[\I\{ \cS\cap \cR \ne \emptyset \}]  \nonumber \\
	& = \sum_{v \in V } \Pr\{v =\rroot(\cR) \} \cdot  \E[\I\{ \cS \cap \cR \ne \emptyset \} \mid v =\rroot(\cR)] 
	\nonumber \\
	& = \frac{1}{n} \cdot \sum_{v \in V} \E[\I\{ \cS\cap \cR_v \ne \emptyset \} ] \nonumber  \\
	& = \frac{1}{n} \cdot \sum_{v \in V} \E\left[\I\{ v \in \bigcup_{i=1} \Gamma(L_i, S_i)\} \right] \label{eq:crimm}  \\
	& = \frac{1}{n} \cdot \E \left[ \sum_{v \in V} \I\{ v \in \bigcup_{i=1} \Gamma(L_i, S_i)\} \right] \nonumber \\
	& = \frac{1}{n} \cdot \E\left[\bigcup_{i=1} \Gamma(L_i, S_i)\right] \nonumber \\
	& = \frac{1}{n} \cdot \rho(\cS), \nonumber 
	\end{align}
	where Eq.~\eqref{eq:crimm} is based on the equivalence between multi round RR sets and all $T$ rounds live-edge graphs ,
	and the expectation from this point on is taken over the random live-edge graphs $L_1, L_2, \ldots, L_T$. 
\end{proof}

\begin{lemma} \label{lem:crtime}
	\alg{\CRNS} returns the multi-round sets $\cS$ which covers the multi-round RR set with $\frac{1}{2}$-approximate optimal solution. The execution time of \alg{\CRIMM} is $O(T^2kn + \sum_{\cR \in \cM} |\cR|)$.
\end{lemma}

\begin{proof}
	\alg{\CRIMM} is the multi-round RR set implementation of the \alg{\CRGreedy}.
	Due to the property of partition matroid, we know the greedy soluction returns a $\frac{1}{2}$-approximate solution.
	
	About time complexity, in each iteration, it takes $O(\sum_{\cR \in \cM} |\cM|)$ to generate the count array and RR set link for each node in lines \ref{line:crgen1} and \ref{line:crgen2} of \alg{\CRNS}.
	In next $Tk$ round, every time it takes $O(Tn)$ for every avaliable seed with the maximum influence, and so the total time is $O(T^2kn)$.
	The rest part of \alg{\CRNS} adjusts the $c[(u, t)]$ and $covered[\cR]$ (from lines \ref{line:crfor2} to \ref{line:crfor3}).
	While $v$ is a chosen seed, each multi-round RR set is traversed at most once for updating both $c[(u, t)]$ and $covered[\cR]$, so it takes $O(\sum_{\cR \in \cM} |\cR|)$ for the rest part.
	So, the total execution time of $\alg{\CRNS}$ is $O(T^2kn + \sum_{\cR \in \cM} |\cR|)$.
\end{proof}

Assume $\hat{\rho}(\cS, \omega)$ is the random esitimation of $\rho(\cS)$, and $\omega$ is a random sample in space $\Omega$. Let $\cS^*$ is the optimal solution of $\rho(\cS)$, so $\cS^* = \argmax_{\cS \subset \cV, |S_i| \leq k, i \in [T], S_i \in \cS } \rho(\cS)$, and  $\OPT = \rho(\cS^*)$. For a random sample $\omega \in \Omega$, let $\hat{\cS}^g(\omega)$ be the solution of the greedy algorithm of $\hat{\rho}(\cdot, \omega)$. For a $\varepsilon > 0$, we have a bad solution $\cS$ under $\varepsilon$,if $\rho(\cS)< (1- 1/e - \varepsilon) \cdot \OPT$. Next, we give the therotical guarantee about the approximate solution of the greedy alogrithm. 

\begin{lemma} \label{lem:crerror}
	For any $\varepsilon$, $\varepsilon_1 \in (0, \varepsilon/(1/2))$, and any $\delta_1, \delta_2 > 0$,
	if (a) $\Pr\{\hat{\rho}(\cS^*, \omega)) \geq (1 - \varepsilon_1) \cdot \OPT\} \geq 1 - \delta_1$.
	(b) $\Pr\{\hat{\rho}(\cS, \omega)) \geq \frac{1}{2}(1 - \varepsilon_1) \cdot \OPT\} \leq \delta_2 / {\binom{n}{k}}^T$;
	(c) For any $\omega \in \Omega$, $\hat{\rho}(\cS, \omega)$  is monotone and submudolar, and the greedy soluction $\cS^o(\omega)$ is a $(\frac{1}{2} - \varepsilon)$-approximate solution with at least $1-\delta_1-\delta_2$ probability.
\end{lemma}

\begin{proof}
	Since $\hat{\rho}(\cS, \omega)$ holds monotone and submudolar, by Lemma \ref{lem:crerror}, we know
	\begin{equation}
		\hat{\rho}(\cS^o(\omega), \omega) \geq \frac{1}{2} \cdot \hat{\rho}(\cS^*, \omega) \nonumber.
	\end{equation}
	Due to condition $(a)$, we know
	\begin{equation}
		\hat{\rho}(\cS^o(\omega), \omega) \geq \frac{1}{2} \cdot \hat{\rho}(\cS^*) \geq \frac{1}{2}(1 - \varepsilon_1) \cdot \OPT \nonumber,
	\end{equation} 
	with at least $1-\delta_1$ probability.
	
	Due to condition $(b)$, we know there is a bad $\cS$ causing $\hat{\rho}(\cS, \omega)) \geq \frac{1}{2}(1 - \varepsilon_1) \cdot \OPT$ with at most $\delta_2$ probability.
	Because of at most ${\binom{n}{k}}^T$ bad samples for $k$ seeds in $T$ rounds,
	the probability of each bad $S$ is at most $\delta_2/{\binom{n}{k}}^T$.
	By using union bound, we know ${\rho}(\cS^o(\omega)) \geq \frac{1}{2}(1 - \varepsilon_1) \cdot \OPT$ with at least $1-\delta_1-\delta_2$ probability.
	
\end{proof}

In multi-round RR-set-based algorithms, the random sample is a multi-round RR set $\cR$ by random generation.
Let $\theta = |\cM|$ is the number of multi-round RR sets, and then $\hat{\rho}(\cS, \omega) = \hat{\rho}(\cS, \cM) = n \cdot \kappa(\cS \cap \cM) / \theta$ which returns the number of $\cR \in \cM$ covered by $\cS$.
The output of \alg{\CRNS} is $\hat{\cS}^g = \hat{\cS}^g(\cM)$.

\begin{definition}(Martingale) \label{def:martingale}
	A sequence of random variables $X_1, X_2, X_3, \ldots$ is a martingale, for all $i \geq 1$, if and only if $\E[|X_i|] < +\infty$ and $\E[X_{i+1} \mid X_1, X_2, \ldots , X_{i}] = X_{i}$.
\end{definition}

\begin{lemma} \label{lem:martingale}
	Suppose $X_1, X_2, X_3, \ldots, X_{t}$ are random variable between $[0, 1]$.
	There is a $\mu \in [0, 1]$ that $\E[X_{i} \mid X_1, X_2, \ldots , X_{i-1}] = \mu$ for all $i \in [t]$.
	Let $Z_i = \sum_{j=1}^i(X_j - \mu)$, a sequence $Z_1, Z_2,  \ldots, X_t$ is a martingale.
\end{lemma}

\begin{proof}
	This lemma can be proved directly. First, $\E[Z_{i+1} \mid Z_1, Z_2, \ldots , Z_{i}] = \E[\sum_{j=1}^{i+1}(X_j - \mu \mid Z_1, Z_2, \ldots , Z_{i}]  = \sum_{j=1}^{i}(X_j - \mu) + \E[X_{i+1} - \mu \mid Z_1, Z_2, \ldots , Z_{i}]  =Z_{i}$. 
	Second, $\E[|Z_i|] \leq 2i \leq +\infty$. So, a sequence $Z_1, Z_2,  \ldots, X_t$ is a martingale.
	
\end{proof}

\begin{corollary} \label{cor:ineq}
	Assume $X_1, X_2, X_3, \ldots, X_{t}$ are random variable between $[0, 1]$.
	There is a $\mu \in [0, 1]$ that $\E[X_{i} \mid X_1, X_2, \ldots , X_{i-1}] = \mu$ for all $i \in [t]$.
	Let $Y = \sum_{i=1}^t X_i$, for any $\gamma > 0$,
	\begin{equation}
		\Pr\{Y - t\mu \geq \gamma \cdot t\mu\}\leq \exp\left(-\frac{\gamma^2}{2+\frac{2}{3}\gamma}t\mu\right) \nonumber.
	\end{equation}
	For any $ 0< \gamma < 1$,
	\begin{equation}
	\Pr\{Y - t\mu \leq - \gamma \cdot t\mu\}\leq \exp\left(-\frac{\gamma^2}{2}t\mu\right) \nonumber.
	\end{equation}
\end{corollary}

\begin{lemma} \label{lem:crtheta}
	Assume a set $\cM = \{\cR_1, \cR_2, \ldots, \cR_{\theta}\}$ satisfies the martingale condition.
	For any $\varepsilon > 0$, $\varepsilon_1 \in (0, \frac{1}{2} \varepsilon)$, and any $\delta_1, \delta_2 > 0$,
	let
	\begin{equation}
		\theta^{(1)} = \frac{2n \cdot \ln(1/\delta_1)}{\OPT \cdot \varepsilon_1^2}, \theta^{(2)} = \frac{n \cdot T\ln(\binom{n}{k}/\delta_2)}{\OPT \cdot (\varepsilon - \frac{1}{2}\varepsilon_1)^2}. \nonumber
	\end{equation}
	If $\theta \geq \theta^{(1)}$, $\Pr\{\hat{\rho}(\cS^*,\cR) \geq (1-\varepsilon_1)\cdot \OPT\} \geq 1 - \delta_1$;
	If $\theta \geq \theta^{(2)}$, for any bad $\cS$ under $\varepsilon$, $\Pr\{\hat{\rho}(\cS,\cR) \geq \frac{1}{2}(1-\varepsilon_1)\cdot \OPT\} \leq \delta_2 / {\binom{n}{k}}^T$;
\end{lemma}

\begin{proof}
	Since $\cM$ satisfies the martingale condition, by using Corollary \ref{cor:ineq}, we have
	\begin{align*}
		& \Pr\{\hat{\rho}(\cS^*,\cR) < (1-\varepsilon_1)\cdot \OPT\} \\
		= &	\Pr\left\{n \cdot \frac{\sum_{i=1}^{\theta}X_i^{\cM}(\cS^*)}{\theta} < (1-\varepsilon_1)\cdot \rho(\cS^*)\right\} \\
		= &	\Pr\left\{\sum_{i=1}^{\theta}X_i^{\cM}(\cS^*) - \theta \cdot \rho(\cS^*) / n< - \varepsilon_1 \left(\theta \cdot\rho(\cS^*)/n\right) \right\} \\
		\leq & \exp\left(- \frac{\varepsilon_1^2}{2} \cdot \theta \cdot \rho(\cS^*) / n \right) \\
		\leq & \exp\left(- \frac{\varepsilon_1^2}{2} \cdot \frac{2n \cdot \ln(1/\delta_1)}{\OPT \cdot \varepsilon_1^2} \cdot \frac{\rho(\cS^*)}{n} \right) = \delta_1.
	\end{align*}
	Let $\varepsilon_2 = \varepsilon - \frac{1}{2} \varepsilon_1$. Let a $\cS$ is bad under $\varepsilon$, thus $\rho(S) < (\frac{1}{2} - \varepsilon) \cdot \OPT$.
	Then, we have
	\begin{align*}
		& \Pr\left\{\hat{\rho}(\cS,\cR) \geq \frac{1}{2}(1-\varepsilon_1)\cdot \OPT\right\} \\
		= &	\Pr\left\{\sum_{i=1}^{\theta}X_i^{\cM}(\cS) - \theta \cdot \rho(\cS)/n \geq \theta/n \cdot \left( \frac{1}{2}(1-\varepsilon_1)\cdot \OPT- \rho(\cS)\right) \right\} \\
		\leq &	\Pr\left\{\sum_{i=1}^{\theta}X_i^{\cM}(\cS) - \theta \cdot \rho(\cS)/n \geq \theta/n \cdot \varepsilon_2 \cdot OPT \right\} \\
		= &	\Pr\left\{\sum_{i=1}^{\theta}X_i^{\cM}(\cS) - \theta \cdot \rho(\cS)/n \geq (\varepsilon_2\cdot \OPT/\rho(\cS)) \cdot \theta \cdot \rho(\cS)/n \right\} \\
		\leq & \exp\left(- \frac{\left(\varepsilon_2\cdot \frac{\OPT}{\rho(\cS)}\right)^2}{2+\frac{2}{3}\left(\varepsilon_2\cdot \frac{\OPT}{\rho(\cS)}\right)} \cdot \theta \cdot \rho(\cS) / n \right) \\
		\leq & \exp\left(- \frac{\left(\varepsilon_2\cdot \OPT\right)^2}{2\rho(\cS)+\frac{2}{3}\left(\varepsilon_2\cdot \OPT\right)} \cdot \theta \cdot 1 / n \right) \\
		\leq & \exp\left(- \frac{\left(\varepsilon_2\cdot \OPT\right)^2}{2(\frac{1}{2} - \varepsilon) \cdot \OPT+\frac{2}{3}\left(\varepsilon_2\cdot \OPT\right)} \cdot \theta \cdot 1 / n \right) \\
		\leq & \exp\left(- \frac{(\varepsilon - \frac{1}{2} \varepsilon_1)^2\cdot \OPT}{n} \cdot \theta \right) \\
		\leq & \exp\left(- \frac{(\varepsilon - \frac{1}{2} \varepsilon_1)^2\cdot \OPT}{n} \cdot \frac{n \cdot \ln({\binom{n}{k}}^T/\delta_2)}{\OPT \cdot (\varepsilon - \frac{1}{2}\varepsilon_1)^2} \right) = \delta_2/{\binom{n}{k}}^T.
	\end{align*}
\end{proof}

\begin{corollary} \label{cor:alphabeta}
	Assume $\cM$ satisfies the matingale condition. We set the parameters $\delta_1 = \delta_2 = 1/(4n^\ell)$, and $\varepsilon_1 = \varepsilon \cdot \frac{\alpha}{\frac{1}{2}\alpha+\beta}$ in Lemma \ref{lem:crtheta}, where
	\begin{equation}
		\alpha = \sqrt{l \ln{n}+ \ln{4}}, \beta = \sqrt{\frac{1}{2} \cdot (\ln {\binom{n}{k}}^T+ \alpha^2 )} \nonumber.
	\end{equation}
	While $\theta \geq \frac{2n\cdot (\frac{1}{2}\alpha + \beta)^2}{\OPT \cdot \varepsilon^2}$, the \alg{\CRIMM} returns $(\frac{1}{2} - \varepsilon)$-approximate solution with at least $1- 1 / (2n^\ell)$ probability.
\end{corollary}

\begin{proof}
	While $\delta_1 = \delta_2 = 1/(4n^\ell)$, we set the $\theta^{(1)} = \theta^{(2)}$ in Lemma \ref{lem:crtheta},
	then we can find
	$\theta^{(1)} = \theta^{(2)} = \frac{2n\cdot (\frac{1}{2}\alpha + \beta)^2}{\OPT \cdot \varepsilon^2}$, while $\varepsilon_1 = \varepsilon \cdot \frac{\alpha}{\frac{1}{2}\alpha+\beta}$ by calculation. The condition (a) (b) of Lemma \ref{lem:crtheta} is satisfied under this setting. For condition (c), it is very easily to prove the monotone and submodular of $\hat{\rho}(\cS, \cM) = n \cdot \kappa(\cS, \cM)/\theta$. Thus, by Lemma \ref{lem:crerror}, we know that \alg{\CRNS} returns $\hat{\cS}^g$ which is a $1 - 1/e -\varepsilon$-approximate solution with at least $1-1/(2n^\ell)$.
\end{proof}

\begin{lemma} \label{lem:crmatingale}
	Let $\cM = \{\cR_1, \cR_2, \ldots, \cR_t\}$ is the multi-round RR sets in \alg{\CRIMM}. For all subset $\cS \subseteq \cV$, for all $i \in [t]$, $\E[X_i^{\cM}(\cS) \mid X_1^{\cM}(\cS), X_2^{\cM}(\cS), \ldots, X_{i-1}^{\cM}(\cS)] = \rho(\cS)/n$, where $X_i^{\cM}(\cS) = \I\{\cS \cap \cR_i \ne \emptyset\}$. So $\cM$ satisfies all matingale conditions.
\end{lemma}

\begin{proof}
	$X_i^{\cM}(\cS)$ is not indepent with $X_1^{\cM}(\cS), \ldots, X_{i-1}^{\cM}(\cS)$, because the generation of $\cR_i$ is decided by $\cR_1, \ldots, \cR_{i-1}$. However, while the decision of $\cR_i$ generation is made, $\cR_i$ is generated by random pick the root $v \in V$ and do reverese propagation in each round for multi-round RR set, which is  totaly independent and not relevant with  $\cR_1, \ldots, \cR_{i-1}$.
	By Lemma \ref{lem:crimm}, we have $\E[X_i^{\cM}(\cS) \mid X_1^{\cM}(\cS), X_2^{\cM}(\cS), \ldots, X_{i-1}^{\cM}(\cS)] = \rho(\cS)/n$. By Lemma \ref{lem:martingale}, $\cM$ satisfies all martingale conditions.
\end{proof}

\begin{lemma} \label{lem:half}
	For each $i = 1,2, \ldots, \left \lfloor \log_2n \right \rfloor -1,$ (1) If $x_i = n / 2^i > \OPT$, $\hat{\rho}(\cS, \cM_i) \geq (1 + \varepsilon') \cdot x_i$ with at most $\frac{1}{2^\ell\log_2{n}}$ probability; (2) If $x_i = n/2^i \leq \OPT$, $\hat{\rho}(\cS, \cM_i) \geq (1 + \varepsilon') \cdot \OPT$ with at most $\frac{1}{2^\ell\log_2{n}}$ probability.
\end{lemma}

\begin{proof}
	Let $\cM = \{\cR_1, \cR_2, \ldots, \cR_{\theta_i}\}$. We know $\hat{\rho}(\cS_i, \cM_i) = n \cdot \sum_{j=1}^{\theta_i}X_j^{\cM_i}(\cS_i)/\theta_i$. By Lemma \ref{lem:crmatingale}, we know $\cM$ satifies the martingale conditions, so we can use Corollary \ref{cor:ineq} for rest proof.
	First, we need to prove (1) Suppose $x_i \geq = n/2^i > \OPT$. For any $\cS$, we have
	\begin{align*}
	&\Pr\{\hat{\rho}(\cS,\cM_i) \geq (1+\varepsilon')\cdot x_i\} \\
	&= 	\Pr\left\{n \cdot \sum_{j=1}^{\theta_i}X_j^{\cM_i}(\cS_i)/\theta_i \geq (1+\varepsilon')\cdot x_i\right\} \\
	&= 	\Pr\left\{\sum_{j=1}^{\theta_i}X_j^{\cM_i}(\cS_i) - \theta_i \cdot \rho(\cS) /n \geq \theta_i \cdot (1+\varepsilon')\cdot x_i/n - \theta_i \cdot \rho(\cS) /n\right\}\\
	& \leq \Pr\left\{\sum_{j=1}^{\theta_i}X_j^{\cM_i}(\cS_i) - \theta_i \cdot \rho(\cS) /n \geq \theta_i \cdot \varepsilon' \cdot x_i/n \right\}\\
	& = \Pr\left\{\sum_{j=1}^{\theta_i}X_j^{\cM_i}(\cS_i) - \theta_i \cdot \rho(\cS) /n \geq (\varepsilon' \cdot x_i/\rho(\cS)) \cdot (\theta_i \cdot  \rho(\cS)/n) \right\}\\
	& \leq \exp\left(-\frac{(\varepsilon'\cdot x_i/\rho(\cS))^2}{2+\frac{2}{3}(\varepsilon'\cdot x_i/ \rho(\cS))}(\theta_i\cdot \rho(\cS)/n)\right) \\
	&= \exp\left(-\frac{(\varepsilon'\cdot x_i)^2}{2\rho(\cS)+\frac{2}{3}(\varepsilon'\cdot x_i)}(\theta_i/n)\right)\\
	& \leq \exp\left(-\frac{{\varepsilon'}^2\cdot x_i}{2+\frac{2}{3}\varepsilon'}(\theta_i/n)\right)\\
\end{align*}
\begin{align*}
	& \leq \exp\left(-\frac{{\varepsilon'}^2\cdot x_i}{n(2+\frac{2}{3}\varepsilon')} \cdot \frac{n\cdot(2+\frac{2}{3}\varepsilon')\cdot(T\ln{\binom{n}{k}}+\ell \ln{n}+\ln{2}+\ln{\log_2{n}})}{\varepsilon'\cdot x_i}\right)\\
	& = \frac{1}{2{\binom{n}{k}}^Tn^\ell\log_2{n}}
	\end{align*}
For the output $\cS_i$ in $i$-th round, we can use the union bound to get
\begin{align*}
	&\Pr\{\hat{\rho}(\cS_i,\cM_i) \geq (1+\varepsilon')\cdot x_i\} \\
	&= 	\sum_{S_t \in 2^{|V|}, S_t \in \cS, |S_t| = k, |\cS|/|S_t| = T, t\in[T]}\Pr\{\hat{\rho}(\cS_i,\cM_i) \geq (1+\varepsilon')\cdot x_i \wedge \\
	& \qquad  \qquad \qquad \qquad \qquad \qquad \qquad \qquad \quad(\cS_i = \cS)\} \\
	&\leq 	\sum_{S_t \in 2^{|V|}, S_t \in \cS, |S_t| = k, |\cS|/|S_t| = T, t\in[T]}\Pr\{\hat{\rho}(\cS_i,\cM_i) \geq (1+\varepsilon')\cdot x_i\} \\
	&\leq 	\sum_{S_t \in 2^{|V|}, S_t \in \cS, |S_t| = k, |\cS|/|S_t| = T, t\in[T]}\frac{1}{2{\binom{n}{k}}^Tn^\ell\log_2{n}} =  \frac{1}{2n^\ell\log_2{n}}
\end{align*}

Similary to prove (1), we show (2) that suppose $x_i \leq n/2^i \cdot \OPT$, for any multi-round set $\cS$, we have
	\begin{align*}
	&\Pr\{\hat{\rho}(\cS,\cM_i) \geq (1+\varepsilon')\cdot \OPT\} \\
	&= 	\Pr\left\{n \cdot \sum_{j=1}^{\theta_i}X_j^{\cM_i}(\cS_i)/\theta_i \geq (1+\varepsilon')\cdot \OPT\right\} \\
	&= 	\Pr\left\{\sum_{j=1}^{\theta_i}X_j^{\cM_i}(\cS_i) - \theta_i \cdot \rho(\cS) /n \geq \theta_i \cdot (1+\varepsilon')\cdot \OPT/n - \theta_i \cdot \rho(\cS) /n\right\}\\
	& \leq \Pr\left\{\sum_{j=1}^{\theta_i}X_j^{\cM_i}(\cS_i) - \theta_i \cdot \rho(\cS) /n \geq \theta_i \cdot \varepsilon' \cdot \OPT/n \right\}\\
	& = \Pr\left\{\sum_{j=1}^{\theta_i}X_j^{\cM_i}(\cS_i) - \theta_i \cdot \rho(\cS) /n \geq (\varepsilon' \cdot \OPT/\rho(\cS)) \cdot (\theta_i \cdot  \rho(\cS)/n) \right\}\\
	& \leq \exp\left(-\frac{(\varepsilon'\cdot \OPT/\rho(\cS))^2}{2+\frac{2}{3}(\varepsilon'\cdot \OPT/ \rho(\cS))}(\theta_i\cdot \rho(\cS)/n)\right) \\
	&= \exp\left(-\frac{(\varepsilon'\cdot \OPT)^2}{2\rho(\cS)+\frac{2}{3}(\varepsilon'\cdot \OPT)}(\theta_i/n)\right)\\
	& \leq \exp\left(-\frac{{\varepsilon'}^2\cdot \OPT}{2+\frac{2}{3}\varepsilon'}(\theta_i/n)\right)\\
	& \leq \exp\left(-\frac{{\varepsilon'}^2\cdot x_i}{2+\frac{2}{3}\varepsilon'}(\theta_i/n)\right)\\
	& \leq \exp\left(-\frac{{\varepsilon'}^2\cdot x_i}{n(2+\frac{2}{3}\varepsilon')} \cdot \right. \\ 
	& \qquad \qquad \left. \frac{n\cdot(2+\frac{2}{3}\varepsilon')\cdot(T\ln{\binom{n}{k}}+\ell \ln{n}+\ln{2}+\ln{\log_2{n}})}{\varepsilon'\cdot x_i}\right)\\
	& = \frac{1}{2{\binom{n}{k}}^Tn^\ell\log_2{n}}
\end{align*}
By using union bound again, we have $\Pr\{\hat{\rho}(\cS,\cM_i) \geq (1+\varepsilon')\cdot \OPT\} \leq \frac{1}{2n^\ell\log_2{n}}$
\end{proof}

\begin{lemma} \label{lem:lb}
	The lower boud $LB$ in \alg{\CRIMM} is smaller than $\OPT$ with at least $1 - \frac{1}{2n^\ell}$ probability.
\end{lemma}

\begin{proof}
Let $LB_i = \hat{\rho(\cS_i, \cM_i)}$. First, we discuss the situation while $\OPT \geq x_{\lfloor \log_2n  \rfloor -1}$. Let $i$ to be smallest round when $\OPT \geq x_i$, and then for all $i' \leq i-1$, $\OPT < x_{i'}$. By Lemma \ref{lem:half} (1),$\hat{\rho}(\cS_{i'}, \cM_{i'}) \geq (1+\varepsilon') \cdot x_{i'}$ with at most $\frac{1}{2n^\ell\log_2{n}}$ probability. By using union bound, we know that for $i' \leq i-1$, $\hat{\rho}(\cS_{i'}, \cM_{i'}) < (1+\varepsilon') \cdot x_{i'}$ with at least $1- \frac{i-1}{2n^\ell\log_2{n}}$. It means the for loop would not break out before $i-1$ round with at least $1- \frac{i-1}{2n^\ell\log_2{n}}$ probability, so as to $LB = LB_{I''}$, $I'' \geq i$ or $LB = 1$. While $i'' \geq i$ and $x_{i''} < \OPT$, by Lemma \ref{lem:half} (2), we know that $\hat{\rho}(\cS_i, \cM_i) > (1+\varepsilon') \cdot \OPT$, \em{i.e.} $LB_{I''} > \OPT$, with at most $\frac{1}{2n^\ell\log_2{n}}$ probability. Together, we know that $LB > \OPT$ with at most $\frac{1}{2n^\ell}$ probability.

If $\OPT < x_{\lfloor \log_2n  \rfloor -1}$, similar to above discussion, we also will know that the for loop would not break out in any round with at least $1- \frac{1}{2n^\ell}$ probability. In this case, $LB = 1$, and $LB \leq \OPT$.
\end{proof}

\begin{lemma}(Correctness)\label{lem:crcorrectness}
	For any $\varepsilon > 0$, $\ell > 0$, the output $\cS^o$ of the cross-round algorithm $\alg{\CRIMM}$ is $( \frac{1}{2} - \varepsilon)$-approximate solution with at least $1-\frac{1}{n^\ell}$ probability.
\end{lemma}	

\begin{proof}
	Lemma \ref{lem:lb} proves that $LB \leq \OPT$ with at least $1 - \frac{1}{2n^\ell}$ probability. In this case, we have $\theta \geq \frac{2n(1/2\cdot \alpha+\beta)^2}{\OPT\cdot \varepsilon^2}$. Lemma \ref{lem:crmatingale} proves the multi-round RR sets $\cM$ satisfies the martingale condition. By using Colloary \ref{cor:alphabeta}, while $\theta \geq  \frac{2n(1/2\cdot \alpha+\beta)^2}{\OPT\cdot \varepsilon^2}$, the cross-round node selection \alg{\CRNS} returns a $(1/2-\varepsilon)$-approximate solution. Overall, \alg{\CRIMM} returns $(\frac{1}{2} -\varepsilon)$-approximate solution $\hat{S}^g$  with at least $1 - \frac{1}{n^\ell}$ probability.
\end{proof}

\begin{lemma}\label{lem:crstop}
	For each $i = 1,2, \ldots, \left \lfloor \log_2n \right \rfloor -1$, if $\OPT \geq (1+\varepsilon')^2 \cdot x_i/(1/2)$, $\hat{\rho}(\cS_i, \cM_i) \geq (1/2)\cdot \OPT/(1+\varepsilon')$ and  $\hat{\rho}(\cS_i, \cM_i) \geq (1+\varepsilon') \cdot x_i$ with at least $1 - \frac{1}{2{\binom{n}{k}}^T n^\ell\log_2{n}}$ probability.
\end{lemma}	

\begin{proof}
	Let $\cM_i = \{\cR_1,\cR_2, \ldots, \cR_{\theta_i}\}$. By Lemma \ref{lem:crmatingale}, we know $\cM_i$ satisfies the martingale conditions. Then we can use the Colloary \ref{cor:ineq} for proof. Assume $\OPT \geq (1+\varepsilon')^2 \cdot x_i/(1/2)$, In $i$-th loop. The output $\cS_i$ is the $\frac{1}{2}$-approximate solution of $\cM_i$. Let $\rho(\cS^*) = \OPT$, and then we have $\hat{\rho}(\cS_i, \cM_i) \geq \hat{\rho}(\cS^*, \cM_i) \cdot (1/2)$. Now we have
	 \begin{align*}
	 	&\Pr\{\hat{\rho}(\cS_i,\cM_i) \leq (1/2)\cdot \OPT/(1+\varepsilon')\} \\
	 	& \leq \Pr\{\hat{\rho}(\cS^*,\cM_i) \cdot (1/2) \leq (1/2)\cdot \OPT/(1+\varepsilon')\} \\
	 	&= 	\Pr\left\{n \cdot \sum_{j=1}^{\theta_i}X_j^{\cM_i}(\cS^*)/\theta_i \leq  \OPT/(1+\varepsilon')\right\} \\
	 	&= 	\Pr\left\{\sum_{j=1}^{\theta_i}X_j^{\cM_i}(\cS^*) \leq  \frac{1}{1+\varepsilon'} \cdot \theta_i \cdot \OPT/n\right\} \\
	 	&= 	\Pr\left\{\sum_{j=1}^{\theta_i}X_j^{\cM_i}(\cS^*) - \theta_i \cdot \rho(\cS) /n \leq \frac{1}{1+\varepsilon'} \cdot \theta_i \cdot \OPT/n - \theta_i \cdot \rho(\cS^*) /n\right\}\\
	 \end{align*}	
	 \begin{align*}
	 	& \leq \exp\left(-\frac{\varepsilon'^2}{2(1+{\varepsilon'})^2} \cdot \theta_i\cdot \OPT/n\right) \\
	 	& \leq \exp(-\frac{{\varepsilon'}^2}{2(1+{\varepsilon'}^2)} \cdot \frac{n\cdot(2+\frac{2}{3}\varepsilon')\cdot(T\ln{\binom{n}{k}}+\ell \ln{n}+\ln{2}+\ln{\log_2{n}})}{\varepsilon'\cdot x_i}\\ & \cdot (1+{\varepsilon'})^2 \cdot \frac{x_i}{1-\frac{1}{2}} \cdot \frac{1}{n}) \\
	 	& \leq \frac{1}{2{\binom{n}{k}}^Tn^\ell\log_2{n}}
	 \end{align*}
Thus, $\hat{\rho}(\cS_i, \cM_i) \geq (1/2)\cdot \OPT/(1+\varepsilon')$ with at least $1 - \frac{1}{2{\binom{n}{k}}^T n^\ell\log_2{n}}$ probability. Finally, by putting $\hat{\rho}(\cS_i, \cM_i) \geq (1/2)\cdot \OPT/(1+\varepsilon')$  and $\OPT \geq (1+\varepsilon)^2 \cdot x_i/(1/2)$  together, we have $\hat{\rho}(\cS_i, \cM_i) \geq (1+\varepsilon') \cdot x_i$. 
\end{proof}

For convenience, let us define

\begin{align*}
&\lambda^* = \frac{2n((1/2)\cdot \alpha + \beta)^2}{\varepsilon^2}, \\
& \lambda' = \frac{n(2+ \frac{2}{3}\varepsilon' )\cdot(\ln{\binom{n}{k}}^T+\ell \ln{n}+\ln{2}+\ln{\log_2{n}})}{{\varepsilon'}^2}
\end{align*}

\begin{lemma} \label{lem:crub}
	The lower bound $LB \geq \frac{(1-1/2)}{(1+\varepsilon')^2}\OPT$ with at least $1-\frac{1}{2{\binom{n}{k}}^Tn^\ell\log_2{n}}$ probability. Based on this,
	if k < n for each round, we get the upper bound
	\begin{equation}
	\E[\theta] \leq \frac{4max\{\lambda^*, \lambda'\}\cdot (1 +\varepsilon')^2}{(1-1/2)\cdot \OPT} + 1 =  O\left(\frac{(Tk+\ell)Tn\log{n}}{ \OPT\cdot \varepsilon^2}\right).
	\end{equation} 
\end{lemma}

\begin{proof}
	First, we need to discuss two situations (1) $\OPT < (1+\varepsilon'^2)\cdot x_{\lfloor \log_2n \rfloor -1}/(1-1/2)$.
	Due to $x_{\lfloor \log_2n \rfloor -1} < 4$,  $\OPT < 4(1+\varepsilon'^2)/(1-1/2)$.
	In worst case, $\theta_{\lfloor \log_2n \rfloor -1} = \lceil\lambda' / x_{\lfloor \log_2n \rfloor -1}\rceil$.
	Because of $LB \geq 1$, we have
	\begin{align*}
	\theta &\leq \max\left\{\left\lceil \frac{\lambda^*}{LB} \right\rceil , \theta_{{\lfloor \log_2{n} \rfloor -1}}\right\}\\
	&\leq \max\{\lceil \lambda^*\rceil, \lceil \lambda'\rceil \}
	\leq \left\lceil \frac{4\max\{\lambda^*,\lambda'\}\cdot (1+{\varepsilon'})^2}{(1-1/2)\cdot \OPT} \right\rceil.
	\end{align*}

	Next we discuss the second situation (2) $\OPT \geq (1+\varepsilon'^2)\cdot x_{\lfloor \log_2n \rfloor -1}/(1-1/2)$. Let $i$ to be the minimum loop number, while $\OPT \geq (1+\varepsilon'^2)\cdot x_I/(1-1/2)$,
	and then we have $\OPT < (1+\varepsilon')^2\cdot x_{i-1}/(1-1/2) = 2(1+\varepsilon')^2\cdot x_i/(1-1/2)$.
	$\hat{\rho}(\cS_i, \cM_i) \geq (1+\varepsilon')\cdot x_i$ with at least $1 - \frac{1}{2{\binom{n}{k}}^Tn^\ell\log_2{n}}$ probability. In this case, it means the loop will break out in $i$-th for loop round. Thus, the number of multi-round RR sets is
	\begin{equation}
		\theta_i = \left\lceil \frac{\lambda'}{x_i} \right\rceil \leq \left\lceil \frac{\lambda' \cdot 2(1+\varepsilon')^2}{(1-1/2)\cdot \OPT}\right\rceil 
	\end{equation}
	
	By Lemma \ref{lem:crstop}, we know if the loop break out in $i$-th round, $LB = \hat{\rho}(\cS_i, \cM_i) \geq (1-1/2) \cdot \OPT / (1 + \varepsilon')^2$. So,
	\begin{align*}
		\theta &\leq \max\left\{\left\lceil \frac{\lambda^*}{LB} \right\rceil , \theta_{i}\right\} \\
		&\leq \max \left\{\left\lceil \frac{\lambda^* \cdot (1+\varepsilon')^2}{(1-1/2)\cdot \OPT}\right\rceil, \left\lceil \frac{\lambda' \cdot 2(1+\varepsilon')^2}{(1-1/2)\cdot \OPT}\right\rceil \right\}\\
		&\leq \left\lceil \frac{2\max\{\lambda^*,\lambda'\}\cdot (1+{\varepsilon'})^2}{(1-1/2)\cdot \OPT} \right\rceil.
	\end{align*}
	
	For loop does not break out $i$-th round with at most $\frac{1}{2{\binom{n}{k}}^Tn^\ell\log_2{n}}$ probability. If we consider the worst case, for loop break out at the end (\em{i.e.} $(\lfloor \log_2n \rfloor -1)$-th round), and $LB = 1$.Then we have $\theta \leq \lceil\max\{\lambda^*, \lambda'\}\rceil$. Let put two prossible cases together, we have
		\begin{align*}
		\E[\theta] &\leq \left(1 -\frac{1}{2{\binom{n}{k}}^Tn^\ell\log_2{n}}\right) \cdot \left\lceil \frac{2\max\{\lambda^*,\lambda'\}\cdot (1+{\varepsilon'})^2}{(1-1/2)\cdot \OPT} \right\rceil \\
		&+ \frac{1}{2{\binom{n}{k}}^Tn^\ell\log_2{n}} \cdot \lceil \max\{\lambda^*, \lambda'\} \rceil \\
		& \leq \left(1 -\frac{1}{2{\binom{n}{k}}^Tn^\ell\log_2{n}}\right) \cdot \left( \frac{2\max\{\lambda^*,\lambda'\}\cdot (1+{\varepsilon'})^2}{(1-1/2)\cdot \OPT} + 1 \right)  \\
		&+ \frac{1}{2{\binom{n}{k}}^Tn^\ell\log_2{n}} \cdot  (\max\{\lambda^*, \lambda'\}+1) \\
		& \leq  \frac{3\max\{\lambda^*,\lambda'\}\cdot (1+{\varepsilon'})^2}{(1-1/2)\cdot \OPT} + 1 
	\end{align*}

	Last inequality equation works due to $\OPT \leq n \leq {\binom{n}{k}}^T$. If $k < n$, by putting two above situations together, we have 
	\begin{align*}
		\E[\theta]  \leq \frac{4\max\{\lambda^*,\lambda'\}\cdot (1+{\varepsilon'})^2}{(1-1/2)\cdot \OPT} +1.
	\end{align*}

	Finally, we use the $\lambda^*$ and $\lambda'$ defined before and $\varepsilon' = \sqrt{\varepsilon}$, we have
	\begin{align*}
		\E[\theta]  \leq O\left(\frac{(Tk+\ell)Tn\log{n}}{ \OPT\cdot \varepsilon^2}\right).
	\end{align*}
\end{proof}

\begin{lemma}(Time Complexity) \label{lem:crtimecomplexity}
	While $k < n$, for any $\varepsilon, l > 0$, the total execution time of the \alg{\CRIMM} is $O(T^2(k+l)(m+n)\log{n}/\varepsilon^2).$
\end{lemma}

\begin{proof}
	By Lemma \ref{lem:crtime}, we know the time execution is $O(T^2kn + \sum_{\cR \in \cM} |M|)$ for each iteration. Due to the maximum number of iteration is $\log_2{n} =1$ in \alg{\CRIMM},
	the total execution time is $O(T^2kn\log{n} + \sum_{\cR \in \cM} |M|) = O(T^2kn\log{n} + \E[\theta] \cdot(EPT+1))$, where $EPT$ is the expectation time of multi-round RR set.
	
	By using the new $\alpha, \beta$ and new $\theta$, we recalucalte
	\begin{equation}
		\E[\theta] = O\left(\frac{(Tk + l) Tn \log{n}}{\OPT \cdot \varepsilon^2}\right) \nonumber.
	\end{equation} 
	
	The $EPT$ also can be estimated by a random node $\tilde{v}$,
	\begin{equation}
	EPT = \frac{m}{n} \cdot \E_{\tilde{v}}[\rho(\tilde{v})] \nonumber.
	\end{equation}
	
	Due to $\E_{\tilde{v}}[\rho(\tilde{v})] \leq \OPT$, the total execution time is
	\begin{align*}
	 &O(T^2kn\log{n} + \E[\theta] \cdot(EPT+1))  \\
	 = & O\left(T^2kn\log{n} + \frac{(Tk + l) Tn \log{n}}{\OPT \cdot \varepsilon^2} \cdot ( \frac{m}{n} \cdot \E_{\tilde{v}}[\rho(\tilde{v})] + 1)\right) \\
	 = &  O\left(\frac{T^2(k+l)(m+n)\log{n}}{\varepsilon^2}\right).
	\end{align*}
\end{proof}

\crmrimt*
\begin{proof}
	By Lemma \ref{lem:crcorrectness}, we prove the correctness of the cross-round algorithm \alg{\CRIMM}
	that the output will hold $1/2 -\varepsilon$-approximate solution to multi-round influence maximization.
	By Lemma \ref{lem:crtimecomplexity}, we show the time complexity of the algorithm \alg{\CRIMM} is $O\left(\frac{T^2(k+l)(m+n)\log{n}}{\varepsilon^2}\right)$.
	Thus, the theorem holds.
\end{proof}


\subsection{Proofs for Section \ref{sec:WRIMM} [\WRIIMM]}\label{app:wrimm}

\wrimm*
\begin{proof}
	
	\begin{align}
		& \E[\I\{ (\cS^{t-1} \cap \cR^{t-1} = \emptyset) \wedge (\cS_t \cap \cR_t \ne \emptyset) \} \nonumber \\
		= & \sum_{v \in V } \Pr\{u =\rroot(\cR) \} \nonumber \\
		& \cdot  \E[\I\{ \cS^{t-1} \cap \cR^{t-1} = \emptyset \}  \cdot \I\{ \cS_{t} \cap \cR_{t} \ne \emptyset \} \mid v =\rroot(\cR)]
		\nonumber \\
		= & \frac{1}{n} \cdot \sum_{v \in V} \E[\I\{ \cS^{t-1} \cap \cR^{t-1}_v \ne \emptyset \} \cdot \I\{ \cS_{t} \cap \cR_{v,t} \}] \nonumber  \\
		= & \frac{1}{n} \cdot \sum_{v \in V} \E\left[\I\{ v \notin \bigcup_{i=1}^{t-1} \Gamma(L_i, S_i)\} \cdot \I\{ v \in \Gamma(L_t, S_t)\} \right] \nonumber \\
		= & \frac{1}{n} \cdot \E \left[ \sum_{v \in V} \I\{ v \in \bigcup_{i=1}^{t} \Gamma(L_i, S_i) \setminus \bigcup_{i=1}^{t-1} \Gamma(L_i, S_i) \} \right] \nonumber \\
		= &\frac{1}{n} \cdot \left(\E\left[\bigcup_{i=1}^{t} \Gamma(L_i, S_i)\right] - \E\left[\bigcup_{i=1}^{t-1} \Gamma(L_i, S_i)\right] \right)  \label{eq:wrimm} \\
		= &\frac{1}{n} \cdot \left(\rho(\cS^t) - \rho(\cS^{t-1})\right)\nonumber \\  
		= &\frac{1}{n} \cdot \left(\rho(\cS^{t-1}\cup S_t) - \rho(\cS^{t-1})\right), \nonumber 
	\end{align}
	where Eq.~\eqref{eq:wrimm} is based on the equivalence between multi-round RR sets and the RR sets rooted by the remaining seeds in round $t-1$,
	and the expectation from this point on is taken over the random live-edge graph $L_t$. 
	%
	%
	%
\end{proof}

}
\end{document}